\documentclass[11pt]{article}

\usepackage{docmute}
\usepackage{proof}
\DeclareFixedFont{\ttb}{T1}{txtt}{bx}{n}{9} 
\DeclareFixedFont{\ttm}{T1}{txtt}{m}{n}{9}  

\usepackage{color}
\definecolor{deepblue}{rgb}{0,0,0.5}
\definecolor{deepred}{rgb}{0.6,0,0}
\definecolor{deepgreen}{rgb}{0,0.5,0}
\usepackage{listings}
\usepackage{xifthen}
\ifthenelse{\isundefined{\ismain}}{
  \newcommand{\ismain}{0}
}{}

\usepackage{enumitem}

\usepackage{tikzit}
\input{hypergraph.tikzdefs}

\tikzstyle{hedge}=[fill=white, draw=black, shape=rectangle, rounded corners=2mm, inner sep=0.2mm, outer sep=-2mm, scale=0.8, minimum height=8mm, minimum width=8mm, tikzit category=hypergraph]
\tikzstyle{hedge blue}=[hedge, fill={rgb,255: red,102; green,204; blue,255}, draw=black, shape=rectangle, tikzit category=hypergraph]
\tikzstyle{node}=[fill=black, draw=black, shape=circle, minimum size=1.5mm, inner sep=0mm, tikzit category=hypergraph]
\tikzstyle{red node}=[fill=red, draw=black, shape=circle, minimum size=1.5mm, inner sep=0mm, tikzit category=hypergraph]
\tikzstyle{node highlight}=[fill=black, draw=blue, thick, shape=circle, minimum size=1.5mm, inner sep=0mm, tikzit category=hypergraph]
\tikzstyle{red node highlight}=[fill=red, draw=blue, thick, shape=circle, minimum size=1.5mm, inner sep=0mm, tikzit category=hypergraph]
\tikzstyle{yellow hedge}=[hedge, fill=yellow, draw=black, shape=rectangle, tikzit category=hypergraph]
\tikzstyle{green hedge}=[hedge, fill=green, draw=black, shape=rectangle, tikzit category=hypergraph]
\tikzstyle{small box}=[fill=white, draw=black, shape=rectangle, minimum height=6mm, minimum width=6mm, tikzit category=string diagram]
\tikzstyle{vsmall box}=[fill=black, draw=black, shape=rectangle, minimum height=4mm, minimum width=1mm, tikzit category=string diagram, inner sep=0]
\tikzstyle{medium box}=[fill=white, draw=black, shape=rectangle, minimum height=11mm, minimum width=6mm, tikzit category=string diagram]
\tikzstyle{semilarge box}=[fill=white, draw=black, shape=rectangle, minimum height=16mm, minimum width=6mm, tikzit category=string diagram]
\tikzstyle{large box}=[fill=white, draw=black, shape=rectangle, minimum height=21mm, minimum width=6mm, tikzit category=string diagram]
\tikzstyle{black dot}=[fill=black, draw=black, shape=circle, minimum size=2mm, inner sep=0mm, tikzit category=string diagram]
\tikzstyle{white dot}=[fill=white, draw=black, shape=circle, minimum size=2mm, inner sep=0mm, tikzit category=string diagram]
\tikzstyle{red dot}=[fill=red, draw=black, shape=circle, minimum size=2mm, inner sep=0mm, tikzit category=string diagram]
\tikzstyle{wlabel}=[fill=none, draw=none, shape=rectangle, tikzit category=string diagram, font={\footnotesize}, inner sep=0pt, tikzit fill={rgb,255: red,102; green,204; blue,255}, tikzit draw={rgb,255: red,102; green,204; blue,255}, yshift=0.3mm]
\tikzstyle{BRchange}=[draw=black, shape=diamond, tikzit shape=circle, tikzit fill={rgb,255: red,96; green,0; blue,0}, diamond split part fill={black,red}, inner sep=-5mm, minimum width=2.7mm, minimum height=1.7mm]
\tikzstyle{RBchange}=[draw=black, shape=diamond, tikzit shape=circle, tikzit fill={rgb,255: red,165; green,0; blue,0}, diamond split part fill={red,black}, inner sep=0, minimum width=2.7mm, minimum height=1.7mm]
\tikzstyle{dummy}=[fill=none, draw=none, shape=circle, font={\small}, inner sep=1pt, tikzit draw=blue, tikzit fill=white]
\tikzstyle{node label}=[fill=none, draw=none, shape=rectangle, tikzit fill=cyan, tikzit draw=cyan, font={\scriptsize}, tikzit shape=circle, inner sep=0pt]
\tikzstyle{empty diag}=[fill=white, draw={rgb,255: red,165; green,165; blue,165}, shape=rectangle, minimum size=1.2 cm, dashed, thick]

\tikzstyle{dashed edge}=[-, dashed, very thick]
\tikzstyle{alt sort}=[-, dashed, dash pattern=on 2pt off 0.5pt, thick, draw=red]
\tikzstyle{diredge}=[->, >={Latex[length=1.5mm]}]
\tikzstyle{diredge highlight}=[->, >={Latex[length=1.5mm]}, draw=blue, thick]
\tikzstyle{diredge highlight alt}=[->, >={Latex[length=1.5mm]}, draw=red, thick]
\tikzstyle{boundary frame}=[-, draw={rgb,255: red,170; green,170; blue,255}, dashed, fill={rgb,255: red,238; green,238; blue,255}, thick, dash pattern=on 2pt off 0.5pt]
\tikzstyle{graph frame}=[-, draw={rgb,255: red,191; green,191; blue,191}, dashed, fill={rgb,255: red,238; green,238; blue,238}, thick, dash pattern=on 2pt off 0.5pt]
\tikzstyle{venn}=[-, draw={rgb,255: red,100; green,100; blue,100}, fill={rgb,255: red,238; green,238; blue,238}, thick, opacity=0.5]
\tikzstyle{def sort}=[-]
\tikzstyle{component}=[-, draw=red, thick]
\tikzstyle{map edge}=[{|->}, >=latex, shorten <=0.5mm, shorten >=0.5mm]
\tikzstyle{hypergraph map edge}=[{|->}, draw=red, shorten <=1mm, shorten >=1mm]
\tikzstyle{cdedge}=[->]
\tikzstyle{big cdedge}=[->, very thick, >=latex]
\tikzstyle{pointer edge}=[->, draw=gray, thick]

\usepackage{amsthm,amsmath,amssymb}
 
\newcommand{\ignora}[1]{ }

\usepackage{cleveref} 

\usepackage{microtype}
\setlist[itemize]{noitemsep, topsep=0pt}
\usepackage{xspace}
\usepackage{mathtools}
\usepackage{multicol}
\usepackage[many]{tcolorbox}
\usepackage[all,2cell]{xy}
\CompileMatrices
\UseAllTwocells
\SilentMatrices
\usepackage{graphicx}
\usepackage{float}
\usepackage{fancybox}
\setenumerate{noitemsep,topsep=0pt,parsep=0pt,partopsep=0pt}
\usepackage{comment}
\usepackage{anyfontsize}
 \usepackage{xcolor,colortbl}
\usepackage{multirow}

\def \catC {\mathbb{C}}

\def \catA {\mathbb{A}}

\def \catC {\mathbb{C}}

\def \PROP {\mathsf{PROP}} 

\newcommand{\sg}{\!\lower1pt\hbox{$\includegraphics[width=8pt]{graffles/greenbullet.pdf}$}\!} 
\newcommand{\sr}{\!\lower1pt\hbox{$\includegraphics[width=8pt]{graffles/redbullet.pdf}$}\!} 
\newcommand{\sbl}{\!\lower1pt\hbox{$\includegraphics[width=8pt]{graffles/blackbullet.pdf}$}\!} 

\newcommand{\frob}{\ensuremath{\mathbf{Frob}}\xspace}

\newcommand{\old}[1]{}

\newcommand{\tr}[1]{\xrightarrow{#1}}    
\newcommand{\tl}[1]{\xleftarrow{#1}}    

\newcommand{\dlcorner}{{\ar@{}[dl]|(.8){\text{\large $\urcorner$}}}}
\newcommand{\drcorner}{{\ar@{}[dr]|(.8){\text{\large $\ulcorner$}}}}

\newcommand{\synTosem}[1]{[\! [ #1 ]\! ]}
\newcommand{\frobTosem}[1]{[ #1 ]}
\newcommand{\allTosem}[1]{\langle\! \langle #1 \rangle \! \rangle}

\newcommand{\Cospan}[1]{\mathsf{Csp}(#1)}

\newcommand\symNet{\lower3pt\hbox{$\includegraphics[width=20pt]{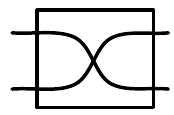}$}}
\newcommand\Idnet{\lower3pt\hbox{$\includegraphics[width=20pt]{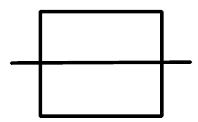}$}}

\newcommand\lccB{\lower5pt\hbox{$\includegraphics[width=25pt]{graffles/rccr.pdf}$}}
\newcommand\rccB{\lower5pt\hbox{$\includegraphics[width=25pt]{graffles/lccl.pdf}$}}
\newcommand\lccn{\lower5pt\hbox{$\includegraphics[width=20pt]{graffles/cup.pdf}$}}
\newcommand\rccn{\lower5pt\hbox{$\includegraphics[width=20pt]{graffles/cap.pdf}$}}
\def \df {\ \ensuremath{:\!\!=}\ }

\DeclareMathOperator{\id}{id}

\newcommand{\Defeq}
 {\stackrel{{def}}{=}}

\newcommand{\stran}{\raise1pt\hbox{$\centerdot$}}

\newcommand{\rring}[1]{\ensuremath{\mathbb{#1}}}

\newcommand{\N}{\rring{N}}

\newcommand{\ra}{\to}

\renewcommand{\emptyset}{\varnothing}

\newcommand{\ladj}[2]{\ar@/^/[#1]^-{#2} \ar@{}[#1]|-%

{\ifthenelse{\equal{#1}{r}}{\bot}{%

{\ifthenelse{\equal{#1}{rr}}{\bot}{%

{\ifthenelse{\equal{#1}{l}}{\top}{%

{\ifthenelse{\equal{#1}{u}}{\dashv}{%

{\vdash}}}}}}}}}}

\newcommand{\radj}[2]{\ar@/^/[#1]^-{#2}}

\newcommand{\radjff}[2]{\ar@{_{(}->}[#1]^{#2}}

\newcommand{\pullbacktop}[4]{%

{#1} \ar@/_/[ddr]_{#4} \ar@/^/[drr]^{#2}%

\ar@{.>}[dr]|-{#3} \\}

\newcommand{\ltsred}[1]
{ \setbox0=\hbox{$\ {}^{#1}\ $}
  \setbox1=\hbox{$\longrightarrow$}
  \loop\setbox1=\hbox{$-$\kern-0.3em\unhbox1}\ifdim\wd1<\wd0\repeat
  \hbox{$\ \ \mathop{\box1}\limits^{#1}\ \ $}
}

\newcommand{\arx}[2]{\!\xymatrix@=15pt{\ar[r]^{{#1}}_{{#2}}&}\!}

\newlength{\mylength}

{\setlength{\fboxsep}{15pt}
\setlength{\mylength}{\linewidth}%
\addtolength{\mylength}{-2\fboxsep}%
\addtolength{\mylength}{-2\fboxrule}%
\Sbox
\minipage{\mylength}%
\setlength{\abovedisplayskip}{0pt}%
\setlength{\belowdisplayskip}{0pt}%
}%
{\endminipage\endSbox
\[\fbox{\TheSbox}\]}

\newcommand{\DCospan}[2]{\mathsf{Csp}_{#1}(#2)}

\newcommand{\syntax}[1]{\mathbf{S}_{#1}}

\newcommand{\FTerm}[1]{\DCospan{D}{\Hyp{{\scriptscriptstyle #1}}}}
\newcommand{\Hyp}[1]{\mathbf{Hyp}_{#1}}

\usepackage{color}
\def\bR{\begin{color}{red}}
\def\bB{\begin{color}{blue}}
\def\bM{\begin{color}{magenta}}
\def\bC{\begin{color}{cyan}}
\def\bW{\begin{color}{white}}
\def\bBl{\begin{color}{black}}
\def\bG{\begin{color}{green}}
\def\bY{\begin{color}{yellow}}
\def\e{\end{color}\xspace}

\newcommand{\cgr}[2][scale=0.45]{\raisebox{0.1em}{\begingroup
\setbox0=\hbox{\includegraphics[#1]{graffles/#2}}%
\parbox{\wd0}{\box0}\endgroup}}

\def \poi {\,\ensuremath{;}\,}
\def \df {\ensuremath{:=}}
\def \tns {\ensuremath{\oplus}}
\def \: {\colon}

\newcommand\Wmult{\itikzfig{Wmult}\xspace}
\newcommand\Wcomult{\itikzfig{Wcomult}\xspace}
\newcommand\Wunit{\itikzfig{Wunit}\xspace}
\newcommand\Wcounit{\itikzfig{Wcounit}\xspace}
\newcommand\Bmult{\itikzfig{Bmult}\xspace}
\newcommand\Bcomult{\itikzfig{Bcomult}\xspace}
\newcommand\Bunit{\itikzfig{Bunit}\xspace}
\newcommand\Bcounit{\itikzfig{Bcounit}\xspace}

\newcommand{\rrule}[2]{\ensuremath{\left\langle #1,#2 \right\rangle}}

\newcommand{\out}[1]{\mathsf{out}(#1)}
\newcommand{\inp}[1]{\mathsf{in}(#1)}

\newcommand{\node}{\lower0pt\hbox{$\includegraphics[width=6pt]{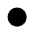}$}}
\newcommand{\hyperedge}{\lower2pt\hbox{$\includegraphics[width=25pt]{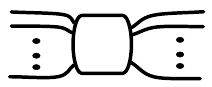}$}}
\newcommand{\ZeronetT}{\lower4pt\hbox{$\includegraphics[width=14pt]{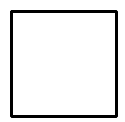}$}}

\newcommand\idncircuit{\lower4pt\hbox{$\includegraphics[width=18pt]{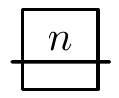}$}}

\newcommand{\NCB}[1]{\ensuremath{\mathbf{NB}_{#1}}\xspace}
\newcommand{\NLC}[1]{\ensuremath{\mathbf{NLC}_{#1}}\xspace}

\usepackage{stmaryrd}

\usepackage{txfonts}
\usepackage{a4wide}

\newtheorem{definition}{Definition}[section]
\newtheorem{proposition}{Proposition}[section]
\newtheorem{theorem}{theorem}[section]
\newtheorem{example}{Example}[section]
\newtheorem{assumption}{Assumption}[section]
\newtheorem{construction}{Construction}[section]
\newtheorem{lemma}{Lemma}[section]
\newtheorem{corollary}{Corollary}[section]
\newtheorem{remark}{Remark}[section]
\newtheorem{notation}{Notation}[section]
\usepackage[colorinlistoftodos,prependcaption,textsize=tiny]{todonotes}

\begin{document}

\title{String Diagram Rewrite Theory III:\\ Confluence with and without Frobenius}
\date{}
\author{Filippo Bonchi \\
University of Pisa
\and
Fabio Gadducci\thanks{Supported from MIUR PRIN 2017FTXR7S ``IT-MaTTerS''.} \\
University of Pisa
\and
Aleks Kissinger \\
University of Oxford
\and
Pawe{\l} Soboci\'nski\thanks{Supported by the ESF funded Estonian IT Academy research measure (project
2014-2020.4.05.19-0001) and the Estonian Research Council grant PRG1210.} \\
Tallinn University of Technology
\and
Fabio Zanasi\thanks{Supported from EPSRC EP/V002376/1.} \\
University College London
}

\maketitle

\begin{abstract}
In this paper we address the problem of proving confluence for string diagram rewriting, which was previously shown
to be characterised combinatorially as double-pushout rewriting with interfaces (DPOI) on (labelled) hypergraphs.
For standard DPO rewriting without interfaces, confluence for terminating rewriting systems is, in general, undecidable.
Nevertheless, we show here that confluence for DPOI, and hence string diagram rewriting, \emph{is} decidable.
We apply this result to give effective procedures for deciding local confluence of symmetric monoidal theories with
and without Frobenius structure by critical pair analysis. For the latter, we introduce the new notion of path joinability
for critical pairs, which enables finitely many joins of a critical pair to be lifted to an arbitrary context in spite of the
strong non-local constraints placed on rewriting in a generic symmetric monoidal theory.
\end{abstract}


\maketitle


\section{Introduction}

This is the final installment in a series of three papers developing the rewriting theory of string diagrams. String diagrams are a practical and visually intuitive language for presenting compositions of morphisms in a symmetric monoidal category. These are particularly useful for expressing \textit{symmetric monoidal theories} (SMTs), which enable one to present symmetric monoidal categories by generators and relations, strictly generalising algebraic theories. In Part I~\cite{BGKSZ-partone}, we showed that when an SMT contains a Frobenius algebra, string diagrams are in one-to-one correspondence with (labelled) hypergraphs, and equational reasoning corresponds precisely to hypergraph rewriting. In Part II~\cite{BGKSZ-parttwo}, we extended this representation to all SMTs, at the expense of requiring certain restrictions on which hypergraphs and hypergraph rewrites are allowed.

In this paper, we address one of the pillars of rewriting theory: The question of \textit{confluence} for string diagram rewriting systems. For \emph{term} rewriting, both confluence \cite{bauer1984finite} and termination \cite{huet1978uniform} are, in general, undecidable. However, for term rewriting systems known to be terminating, confluence is decidable. The key, celebrated property observed by Knuth and Bendix~\cite{knuth1970simple} is that a terminating system is confluent exactly when \emph{all its critical pairs are joinable}.

Since diagrams can be represented combinatorially as hypergraphs, it stands to reason that we should turn to the
literature on \textit{graph rewriting theory} to find answers about confluence.
Here, unfortunately, the status of confluence is less certain because established properties of critical pair analysis fail: Plump~\cite{Plump1993}, working in the framework of the double-pushout (DPO) graph rewriting~\cite{Ehrig1976},
showed that joinability of critical pairs \emph{does not} entail confluence, and even worse: that confluence of
terminating DPO rewriting systems is, in general, undecidable.

Thankfully, in the case of string diagrams, compositionality comes to the rescue. Unlike the graphs considered by Plump,
string diagrams have a natural notion of an interface, namely the inputs and outputs of the diagram that represent the
domain and codomain of a morphism in a symmetric monoidal category. Consequently, the appropriate notion of rewriting
for string diagrams should preserve that interface. This is why in the first two parts of this series we have formalised string
diagram rewriting using \textit{DPO with interfaces} (DPOI).

The idea of performing rewrites that respect an interface is not new, and has emerged in several research threads,
including rewriting with borrowed contexts \cite{Ehrig2004}, 
encodings of process calculi \cite{Gadducci07,BonchiGK09}, 
and connecting DPO rewriting systems with computads in cospans categories \cite{Gadducci1998,Sassone2005a}.
Our key observation is that for DPOI rewriting, the Knuth-Bendix property is saved: Confluence of a
terminating DPOI system can be decided by checking whether its critical pairs are joinable.

In fact, this result is more general than our particular encoding of string diagrams as hypergraphs: Under some mild assumptions related to the computability of individual rewriting steps, our result holds for DPOI rewriting in generic
adhesive categories~\cite{Lack2005}.

Our results do not falsify Plump's: In DPOI rewriting one rewrites morphisms $G \leftarrow J$, thought of as a graph
(or graph-like object) $G$ with a fixed interface $J$. The latter allows to consider $G$ in a larger context, where
$J$ acts as the ``glue'' between $G$ and its context. This is analogous to how variables allow a single term to apply
to a variety of contexts via substitution.
In the light of our analysis, Plump's result states that it is undecidable to check whether rewriting is confluent for all morphisms $G \leftarrow 0$, i.e., objects with an empty interface. Intuitively, the failure of Knuth-Bendix for such morphisms is due to the loss of expressive power of critical pairs, if deprived of an interface.

This reveals an attractive analogy with term rewriting: Morphisms $G \leftarrow 0$ -- representing graphs that can be
only trivially attached to other graphs, since they have an empty interface -- are akin to \emph{ground terms}, that cannot be extended since they have no variables. The property that Plump showed to be undecidable should be compared to \emph{ground confluence} for term rewriting \cite{padawitz1980new}, i.e., confluence with respect to all ground terms. And in fact, this property is undecidable for terminating term rewriting systems \cite{kapur1990ground}. Summarising, for both term and DPOI rewriting, confluence of terminating rewriting systems is decidable, while ground confluence is not.

\vspace{-.35cm}
\begin{center}
\setlength\tabcolsep{5pt}
\begin{tabular}{c|c|c}
& Terminating term rewriting system & Terminating DPO system \\ \hline
Ground confluence & undecidable (Kapur et al.~\cite{kapur1990ground}) & undecidable (Plump~\cite{Plump1993}) \\ \hline
Confluence & decidable (Knuth and Bendix~\cite{knuth1970simple}) & decidable (this paper) \\
\hline
\end{tabular}
\end{center}

We can apply this result about confluence for DPOI rewriting to string diagrams. This problem is known to be particularly challenging, for example a directed form of the Yang-Baxter equation generates infinitely many critical pairs \cite{Lafont2003,Mimram14}.

We show that this issue can be avoided by using DPOI rewriting, and that confluence is decidable in many cases by checking only finitely many critical pairs. The generic result for DPOI rewriting applies essentially verbatim for SMTs containing Frobenius structure that we considered in Part I.

For generic SMTs that do not necessarily have Frobenius structure, the story is is a bit more nuanced. It was shown in Part II that in order to obtain a sound rewriting theory for generic SMTs and avoid introducing directed cycles in diagrams, one should consider \textit{convex} DPOI rewriting. As we foreshadowed in Part II, this can cause problems for confluence, as convex rewrites can sometimes have unexpected non-local effects. Namely, the validity of a convex match depends on the non-existence of paths from outputs to inputs in the image of a pattern graph in the target. Hence, rule applications that create paths can break convexity of matches elsewhere.

%

In this paper, we provide two solutions to this problem. The first solution solves the problem by putting a strong condition on the rewriting systems considered called \textit{left-connectedness}. This condition essentially requires the left-hand sides of all rules to take a form that guarantees that \textit{any} rewrite is already a convex rewrite, hence it gives a very simple way to sidestep the technical challenges of convex rewriting.
Many interesting rewriting systems arising from SMTs (e.g., \cite{Lafont2003,Ghica13,Fiore2013}), including the aforementioned Yang-Baxter rule, enjoy this property. Amongst these is the rewriting system for \emph{non-commutative bimonoids} that was shown to be terminating in Part II. In this paper, we will apply our technique to show that it is also confluent.

The second technique we provide is the first, to our knowledge, necessary and sufficient condition for checking local confluence of generic SMTs with finitely many critical pairs (although necessity comes with a caveat discussed at the end of Section~\ref{sec:path-extensions}). The key point for such systems is that the presence of paths from outputs to inputs in the context of a critical pair can affect its joinability, so we formally adjoin certain additional generators to a theory, which enable us to check not only the critical pairs themselves, but \textit{path extensions} of the critical pairs, which account for these troublesome paths. Perhaps surprisingly, it is sufficient to check only finitely many of these to guarantee a critical pair is joinable in \textit{any} context. We apply this technique to show confluence of a simple SMT that is not left-connected.

\paragraph*{Synopsis.}
The paper has the following structure. Section~\ref{sec:Background} recalls the basic notions concerning
DPO rewriting for graphs with interfaces (DPOI) and PROP rewriting.
Section~\ref{sec:confluence} presents the main technical result, namely, local confluence for DPOI rewriting.
Thanks to the correspondence results established in the previous papers of the series,
this is exploited to prove local confluence for PROP rewriting with a
Frobenius structure in Section~\ref{sec:confluenceFrob} and for two different
proposals of PROP rewriting without Frobenius in Section~\ref{sec:confluenceSMC}.
Section~\ref{case studies} provide two case studies to support the
relevance of our results, while Section~\ref{conclusions} wraps up the paper with some final considerations.

Much of the content of this article is based on a paper published in the proceedings of
ESOP 2017~\cite{BGKSZ-esop17}.
In addition to updating the paper, extending with more examples and explanation, and adapting for consistency with
the previous two \textit{String Diagram Rewrite Theory} papers, this version goes beyond the conference paper
in two directions. First, it makes precise the distinction between pre-critical and critical pairs, providing an equivalent
of the parallel independence theorem for DPOI rewriting (see Proposition~\ref{confparallelpairs}). Second,
and more substantial, the technique for proving local confluence for generic convex rewriting systems using
formal path extensions (see Section~\ref{sec:path-extensions}) is completely new.
This technique is put to work on a new case study (see Section~\ref{nlc confluence}).

\paragraph*{Related work.}
Confluence is a classical topic for both term and graph rewriting, and it has been studied for quite some time.
The key observation is always the same: Identifying a set of rewrite instances whose check could ensure the
Knuth-Bendix property. Classically, this is the set of critical pairs. For DPO rewriting, despite the
undecidability result recalled before, local confluence has been shown to hold with respect to a stronger notion
of joinability for critical pairs~\cite{Plump1993}, and confluence is decidable whenever all critical pairs satisfy a 
syntactic condition, coverability~\cite{Plump10}. More recently, we mention the work on confluence up-to garbage, whose
intuition is to check if the rewriting system is confluent on a subclass of graphs that are of interest~\cite{0001P20}.
And, on a similar note, the work on confluence for DPO with applications conditions~\cite{EhrigHLOG10} seems
also relevant for our investigation. Indeed, both proposals are reminiscent of our restriction to monogamous acyclic
hypergraphs and convex rewriting, as introduced in Section~\ref{sec:confluenceSMC}, and establishing a precise
correspondence is left  for future work.
Despite the introduction of interfaces, our approach to critical pairs is rather classical. A different proposal concerns
\emph{initial conflicts}, a restricted class of critical pairs that still guarantees the Knuth-Bendix property~\cite{LambersO20}:
Also pursuing the adaptation of this notion in the context of DPOI  is left for future work.
Instead, our same notion of confluence has been studied in~\cite{BrugginkCHK11} in the setting of Milner's reactive systems.
By instantiating Proposition 22 in~\cite{BrugginkCHK11} to the category of input-linear cospans (of hypergraphs) and by using
the results relating borrowed context DPO rewriting with reactive systems over cospans in~\cite{sobocinski:thesis}, one
obtains a variant of our Theorem~\ref{th:locconfl}. One restriction of that approach is that the matches are required to be
mono, which rules out our applications to SMTs.

\section{Background}\label{sec:Background}

\begin{notation}
The composition of arrows $f \colon a \to b,\, g\colon b \to c$ in a category $\catC$
is written as $f\poi g$.
For $\catC$ symmetric monoidal, $\tns$ is its monoidal product and $\sigma_{a,b} \colon a \tns b \to b \tns a$ is the symmetry for
objects $a,b \in \catC$. 
\end{notation}

\subsection{DPO rewriting}\label{ssec:DPOrewriting}
%
%

\subsubsection{Adhesive categories and (typed) hypergraphs.}
In order not to restrict ourselves to any concrete model of graphs, we work with adhesive categories~\cite{Lack2005}. Adhesive categories are relevant because they have well-behaved
pushouts along monomorphisms, and for this reason they are convenient as ambient categories for DPO rewriting.
%
%
%
%

An important
example is the category of finite directed hypergraphs $\Hyp{}$.
%
%
An object $G$ of $\Hyp{}$ is a hypergraph with a finite set of \emph{nodes} $G_\star$
and for each $k,l\in\N$ a finite set of \emph{hyperedges} $G_{k,l}$ with $k$ (ordered) sources and $l$ (ordered) targets, i.e., for each $0\leq i < k$ there is the $i$\textsuperscript{th} source map $s_i\colon G_{k,l}\to G_\star$
and for each $0\leq j < l$ the $j$\textsuperscript{th} target map $t_j \colon G_{k,l}\to G_{\star}$. The arrows of $\Hyp{}$ are homomorphisms: functions $G_\star \to H_\star$
such that for each $k,l$, $G_{k,l}\to H_{k,l}$ they respect the source and target maps in the obvious way. The seasoned reader will recognise $\Hyp{}$ as a presheaf topos, and as such, it is adhesive~\cite{Lack2005}.

We shall visualise hypergraphs as follows: $\node$ is a node and $\hyperedge$ is a hyperedge, with ordered tentacles attached to the left boundary linking to sources and those on the right linking to targets
\ctikzfig{hypergraph-unlab2}

A signature $\Sigma$ consists of a set of \emph{generators} $o\colon n \to m$ with arity $n$ and coarity $m$ where $m,n\in \N$. Any signature $\Sigma$ can be considered as a hypergraph  $G_{\Sigma}$ with a single node, in the obvious way.  We can then express \emph{$\Sigma$-labelled hypergraphs} (briefly, $\Sigma$-hypergraphs)
as the objects of the slice
%
category $\Hyp{} \downarrow G_{\Sigma}$, denoted by $\Hyp{\Sigma}$, which is adhesive, since adhesive categories are closed under slice~\cite{Lack2005}.
$\Sigma$-hypergraphs are drawn by labeling hyperedges with generators in $\Sigma$ 
\ctikzfig{hypergraph-lab2}

\subsubsection{DPO rewriting.}
We recall the DPO approach~\cite{Ehrig1976} to rewriting in an adhesive category $\catC$.
A \emph{DPO rule} is a span $L \xleftarrow{} K \xrightarrow{} R$ in $\catC$.
A \emph{DPO system} $\mathcal{R}$ is a finite set of DPO rules.
Given objects $G$ and $H$ in $\catC$, we say that $G$ \emph{rewrites} into $H$ --notation $G \DPOstep{\mathcal{R}} H$--
if there esist $L \tl{} K  \tr{} R$ in $\mathcal{R}$, object $C$ and morphisms
such that the squares below are pushouts
\begin{equation*}\label{eq:dpo2}
\raise12pt\hbox{$
\xymatrix@R=10pt@C=20pt{
L \ar[d]_m   &  K \ar[d]
\ar@{}[dl]|(.8){\text{\large $\urcorner$}}
\ar@{}[dr]|(.8){\text{\large $\ulcorner$}}
\ar[l] \ar[r]  & R \ar[d] \\
 G &  C \ar[l] \ar[r]  & H }$}
\end{equation*}
A \emph{derivation} from $G$ into $H$ 
is a sequence of such rewriting steps.
The arrow $m \colon L \to G$ is called a \emph{match} of $L$ in $G$. A rule $L \tl{} K  \tr{} R$ is said to be \emph{left-linear} if the morphism $K\to L$ is mono. In this case, the matching $m$ fully determines the graphs $C$ and $H$, i.e., for a fixed rule and a matching there is a unique $H$ such that $G \DPOstep{\mathcal{R}} H$, if it exists.
Here, by unique, we mean unique up-to isomorphism. More generally, the rewriting steps will always be \emph{up-to iso}: in a step $G \DPOstep{\mathcal{R}} H$, $G$ and $H$ should not be thought of as single graphs but rather as equivalence classes of isomorphic graphs.

\paragraph{Undecidability of  confluence.}
In DPO rewriting,
the confluence of terminating systems is not decidable, even if we restrict to left-linear rules.

\begin{theorem}[\cite{Plump1993}]\label{thm:plump}
Confluence of terminating DPO systems over $\Hyp{\Sigma}$ is undecidable.
\end{theorem}

Indeed, critical pair analysis for traditional DPO systems fails: for terminating DPO systems, joinability of critical pairs does not necessarily imply confluence.

\begin{definition}[Pre-critical pair and joinability]\label{defn:DPOcritpair}
 Let $\mathcal{R}$ be a DPO system with rules $L_1 \tl{} K_1 \tr{} R_1$ and $L_2 \tl{} K_2 \tr{} R_2$. Consider two derivations with common source $S$
\begin{equation*}
\xymatrix@R=15pt{
    R_1 \ar[d] & \ar[l] K_1 \dlcorner  \ar[d] \ar[r] & L_1 \ar@{}[dr]|(.8){\text{\large $\ulcorner$}\qquad\quad} \ar@/^5pt/[dr]^{f_1} && \ar@/_5pt/[dl]_{f_2} \ar@{}[dl]|(.8){\qquad\quad\text{\large $\urcorner$}} L_2 & \ar[l] K_2 \ar[d] \drcorner\ar[r] & R_2 \ar[d] \\
    H_1  & \ar[l] C_1 \ar[rr] & & S  & & \ar[ll] C_2  \ar[r] & H_2 \\
    }
\end{equation*}
We say that $ H_1 \LDPOstep{} S\DPOstep{}H_2$ is a \emph{pre-critical pair} if $[f_1,f_2] \colon L_1 + L_2 \to S$ is epi;
it is \emph{joinable} if there exists $W $ such that $H_1 \DPOstep{}^* W {\,\,{}^*\!\!\LDPOstep{}} \, H_2$.
\end{definition}

We use the standard notation $\DPOstep{}^*$ for defining a possibly empty sequence of rewriting steps.
Intuitively, in a pre-critical pair $S$ should not be bigger than $L_1+L_2$.
In a critical pair, $L_1$ and $L_2$ must overlap in $S$,
so that the two rewriting steps do not form a \emph{parallel pair}.

\begin{definition}[Parallel pair]\label{defn:DPOparpair}
 Let $\mathcal{R}$ be a DPO system with rules $L_1 \tl{} K_1 \tr{} R_1$ and $L_2 \tl{} K_2 \tr{} R_2$. Consider two derivations with common source $S$
as in Definition~\ref{defn:DPOcritpair}.
%
We say that $ H_1 \LDPOstep{} S\DPOstep{}H_2$ is a \emph{parallel pair} if there exist $g_1 \colon L_1 \to C_2$ and
$g_2 \colon L_2 \to C_1$ making the diagram below commute
\begin{equation*}
\xymatrix@R=15pt{
    K_1  \ar[d] \ar[r] & L_1 \ar@{}[dr]|(.8){\text{\large $\ulcorner$}\qquad\quad} \ar@/^5pt/[dr]  \ar@/^5pt/[drrr] && \ar@/_5pt/[dl] \ar@{}[dl]|(.8){\qquad\quad\text{\large $\urcorner$}} L_2 \ar@/_5pt/[dlll] & \ar[l] K_2 \ar[d] \\
    C_1 \ar[rr] & & S  & & \ar[ll] C_2 \\
    }
\end{equation*}
\end{definition}

The key result for us is that parallel pairs are always joinable (see e.g.~\cite{CorradiniStaf}), and
a pre-critical pair is thus \emph{critical} if it is not a parallel one.
However, for the purposes of this paper, this restriction is immaterial, and we
will mostly stick to pre-critical pairs in our results, as proofs are less tedious.
For the sake of succinctness, most of the examples will instead display only the critical pairs.
For a pre-critical pair which is also a parallel pair, see for instance the first picture of Section~\ref{sec:confluenceBimonoids}.

A notable feature of DPO rewriting is that, unlike in the case of term rewriting, joinability of all critical pairs is not enough to guarantee confluence, even for a terminating rewriting system.

\begin{example}[\cite{Plump1993}]\label{ex:Plump}
Consider a DPO system $\mathcal{R}$ consisting of the following two rules, where we labeled nodes with numbers in order to make the graph morphisms explicit
\[
  \scalebox{0.75}{\input{./figures/plump-ex1.tikz}}
  \qquad\qquad
  \scalebox{0.75}{\input{./figures/plump-ex2.tikz}}
\]
Amongst the several pre-critical pairs, only the following two have non-trivial overlap
\[
  \scalebox{0.75}{\input{./figures/plump-ex3.tikz}}
  \qquad\qquad
  \scalebox{0.75}{\input{./figures/plump-ex4.tikz}}
\]
Both are obviously joinable. However, $\mathcal{R}$ is not confluent, as witnessed by the following
\[
  \scalebox{0.75}{\input{./figures/plump-ex5.tikz}}
\]
\end{example}

However, this `bug' in DPO rewriting can be fixed by considering graphs with interfaces, and DPO rules that respect the interface.

\subsection{DPO rewriting with interfaces.}



Morphisms $G \tl{}J$ will play a special role in our exposition. When $\catC$ is $\Hyp{\Sigma}$, we will call them \emph{(hyper)graphs with interface}. The intuition is that $G$ is a hypergraph and $J$ is an interface that allows $G$ to be ``glued'' to a context.
Note however that such morphisms are not necessarily mono, even if this will be the case in most of our examples.

Given $G \leftarrow J$ and $H \leftarrow J$ in $\catC$, \emph{$G$ rewrites into $H$ with interface $J$} --notation $(G \tl{} J) \DPOstep{\mathcal{R}} (H \tl{} J)$-- if there exist rule $L \tl{} K  \tr{} R$ in $\mathcal{R}$, object $C$, and morphisms
such that the diagram below commutes and the squares are pushouts

\begin{equation*}\label{eq:dpo2a}
\raise25pt\hbox{$
\xymatrix@R=10pt@C=20pt{
L \ar[d]_m   &  K \ar[d]
\ar@{}[dl]|(.8){\text{\large $\urcorner$}}
\ar@{}[dr]|(.8){\text{\large $\ulcorner$}}
\ar[l] \ar[r]  & R \ar[d] \\
 G &  C \ar[l] \ar[r]  & H \\
&  J \ar[u] \ar[ur]  \ar[ul]
}$}
\end{equation*}
Hence, the interface $J$ is preserved by individual rewriting steps.

When $\catC$ has an initial object $0$ (for instance, in $\Hyp{\Sigma}$ $0$ is the empty hypergraph), ordinary DPO rewriting can be considered as a special case, by taking $J$ to be $0$.

Like for traditional DPO, rewriting steps are modulo isomorphism: $G_1 \leftarrow J: f_1$ and $G_2 \leftarrow J: f_2$ are isomorphic if there is an isomorphism $\varphi \colon G_1 \to G_2$ with $f_1\poi \varphi = f_2$.




\begin{example}\label{ex:againstPlump}
Consider the system $\mathcal{R}$ from Example \ref{ex:Plump} and the graph with interface below
\[ \scalebox{0.65}{\input{./figures/plump-ex-int.tikz}} \]
It can be rewritten in two different ways
\begin{equation}\label{eq:criticalpairplumprevisited}
\scalebox{0.65}{\input{./figures/plump-ex-int2.tikz}}
\end{equation}
which, unlike as in Example~\ref{ex:Plump}, results in two distinct hypergraphs with interface.
Notably, the interface $\{0,1\}$ 
maintains the distinct identities of the two nodes initially connected to the hyperedge labelled by $a$, even after that hyperedge is removed.
Notice that if~\eqref{eq:criticalpairplumprevisited} were considered as a critical pair, it would not be joinable. Hence, the counterexample of Plump~\cite{Plump1993} (Example~\ref{ex:Plump}) would not work. This is the starting observation for our formulation of critical pair analysis: in Section~\ref{sec:confluence} we will introduce pre-critical pairs for rewriting with interfaces and we will show that, as in term rewriting, joinability of pre-critical pairs entails confluence.
\end{example}
%



\subsection{PROP rewriting}\label{ssec:PROPrewriting}

\subsubsection{SMTs and PROPs.} A uniform way to express an algebraic structure within a symmetric monoidal category is with a \emph{symmetric monoidal theory} (SMT). A (one-sorted) SMT is a pair $(\Sigma, \mathcal E)$ where $\Sigma$ is a \emph{signature} defined as in Section~\ref{ssec:DPOrewriting}. The set of $\Sigma$-terms is obtained by combining generators in $\Sigma$, the \emph{unit} $\id \colon 1\to 1$ and the \emph{symmetry} $\sigma_{1,1} \colon 2\to 2$ with $;$ and $\tns$. That means, given $\Sigma$-terms $t \colon k\to l$, $u \colon l\to m$, $v \colon m\to n$, one constructs new $\Sigma$-terms $t \poi u \colon k\to m$ and $t \tns v \colon k+m \to l+n$.
The set $\mathcal E$ of \emph{equations} contains pairs $(t,t')$ of $\Sigma$-terms, with the requirement that $t$ and $t'$ have the same arity and coarity. 

Just as ordinary (cartesian) algebraic theories have a categorical rendition as Lawvere categories~\cite{hyland2007category}, the corresponding linear notion (i.e., in the sense that variables can neither be copied, nor discarded) for SMTs is a PROP~\cite{MacLane1965} (\textbf{pro}\-duct and \textbf{p}ermutation category). A PROP is a symmetric strict monoidal category with objects the natural numbers, where $\tns$ on objects is addition. Morphisms between PROPs are identity-on-objects strict symmetric monoidal functors. PROPs and their morphisms form a category $\PROP$. Any SMT $(\Sigma,\mathcal E)$ freely generates a PROP by letting the arrows $n\to m$ be the $\Sigma$-terms $n\to m$ modulo the laws of symmetric monoidal categories and the (smallest congruence containing the) equations $t=t'$ for any $(t,t')\in \mathcal E$.

We write $\syntax{\Sigma}$ to denote the PROP freely generated by $(\Sigma,\varnothing)$.
There is a graphical representation of the arrows of $\syntax{\Sigma}$ as string diagrams, which we now sketch, referring to~\cite{Selinger2009} for the details. A $\Sigma$-term $n \to m$ is pictured as a box with $n$ ports on the left and $m$ ports on the right, which are ordered and referred to with top-down enumerations $1,\dots,n$ and $1,\dots,m$. Compositions via $\poi$ and $\tns$ are drawn respectively as horizontal and vertical juxtaposition, that means, $t \poi s$ is drawn $\cgr[height=.4cm]{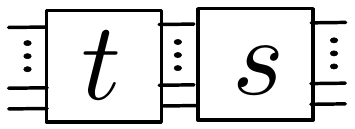}$ and $t \tns s$ is drawn $\cgr[height=.6cm]{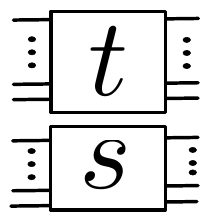}$.
    There are specific diagrams for the $\Sigma$-terms responsible for the symmetries: these are $\id_1 \colon 1 \to 1$, represented as $\Idnet$, the symmetry $\sigma_{1,1} \colon 1+1 \to 1+1$, represented as $\symNet$, and the unit object for $\tns$, that is, $\id_0 \colon 0 \to 0$, whose representation is an empty diagram $\ZeronetT$. Graphical representation for arbitrary identities $\id_n$ and symmetries $\sigma_{n,m}$ are generated using the pasting rules for $\poi$ and $\tns$.
%
It will be sometimes convenient to represent $\id_n$ with the shorthand diagram $\idncircuit$ and, similarly, $t \colon n \to m$ with $\lower5pt\hbox{$\includegraphics[height=15pt ]{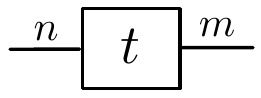}$}$.

\begin{example}\label{exm:props}~
\begin{enumerate}[label=(\alph*)]
\item \label{ex:cmonoids} A basic example is the theory $(\Sigma_{\textbf{Mon}},\mathcal{E}_{\textbf{Mon}})$ of \emph{commutative monoids}. The signature $\Sigma_{\textbf{Mon}}$ contains two generators:  \emph{multiplication} --- which we depict
 $\Bmult \colon 2 \to 1$ --- and \emph{unit}, represented as $\Bunit \colon 0 \to 1$.
 Equations in $\mathcal{E}_{\textbf{Mon}}$ are given in the leftmost column of Figure~\ref{fig:EQFrobenius}: they assert commutativity, associativity and unitality. 

\item \label{ex:frob} An SMT that plays a key role in our exposition is the theory $(\Sigma_{\textbf{Frob}},\mathcal{E}_{\textbf{Frob}})$ of \emph{special Frobenius monoids}.
The signature $\Sigma_{\textbf{Frob}}$ is as follows and $\mathcal{E}_{\textbf{Frob}}$ is depicted in Figure~\ref{fig:EQFrobenius}
\[
 \{ \Bmult \colon 2 \to 1,\, \Bunit \colon 0 \to 1,\, \Bcomult \colon 1\to 2,\, \Bcounit \colon 1\to 0\}
 \]
 \begin{figure}[t]
\[
\includegraphics[height=2cm]{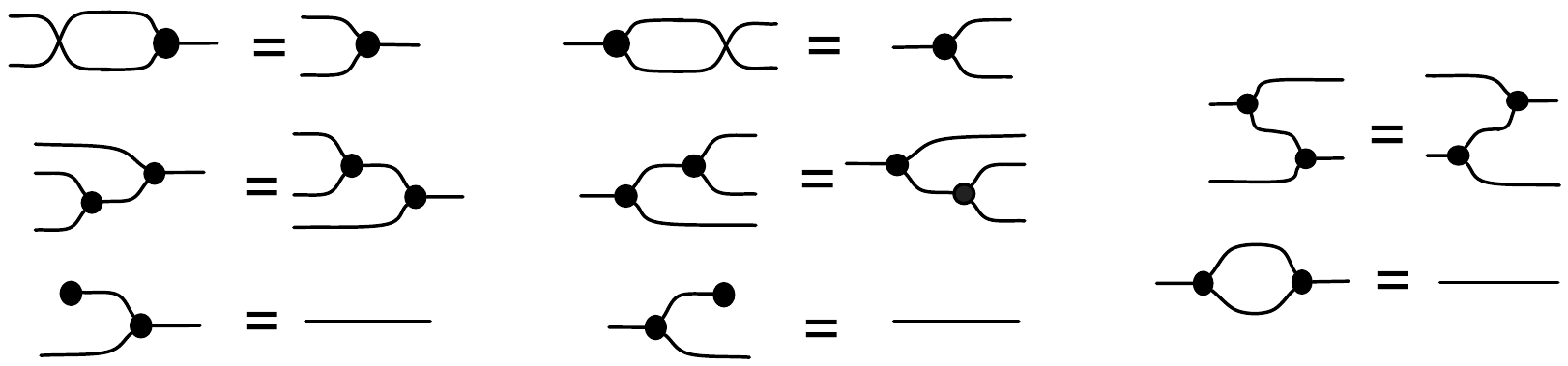}
\]
\caption{The equations $\mathcal{E}_{\textbf{Frob}}$ of special Frobenius monoids.}\label{fig:EQFrobenius}
\end{figure}
%
$E_F$ includes the theory of commutative monoids in the leftmost column.
Dually, the equations in the middle column assert that $\Bcomult$ and $\Bcounit$ form a cocommutative comonoid.
Finally, the two rightmost equations describe an interaction between these two structures.
We call $\frob$ the PROP freely generated by $(\Sigma_{\textbf{Frob}},\mathcal{E}_{\textbf{Frob}})$.

\item \label{ex:ncbialgebras} The theory of \emph{non-commutative bimonoids} has signature $\Sigma_{\textbf{NBiM}}$
\[ \{\Wmult \colon 2 \to 1,\Wunit \colon 0 \to 1,\Wcomult \colon 1\to 2,\Wcounit \colon 1\to 0\}\]
and the following equations $\mathcal{E}_{\textbf{NBiM}}$ \vspace{-1.5cm}
\begin{equation*} \label{eq:ncbialgebralaws}
\lower75pt\hbox{$\includegraphics[height=5cm]{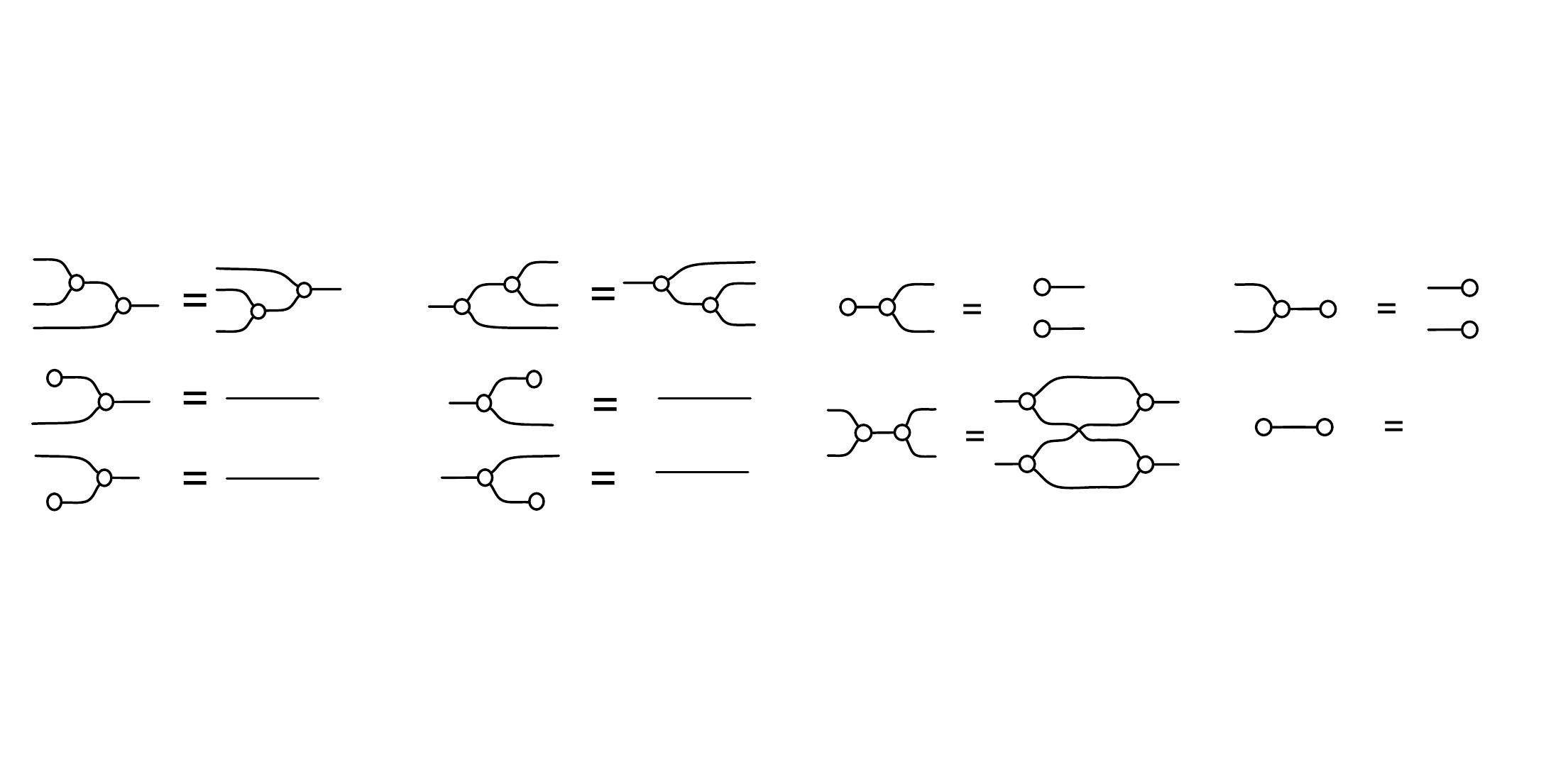}$}\vspace{-1.5cm}
\end{equation*}
We call $\NCB{}$ the PROP freely generated from $(\Sigma_{\textbf{NBiM}},\mathcal{E}_{\textbf{NBiM}})$. 
In~\cite{BGKSZ-parttwo}
we showed that the 
rewriting system that is obtained by orienting the equalities from left to right terminates. In this paper, we will show that is also confluent. For this, it will be convenient to use $\mu,\eta,\nu,\epsilon$, respectively, to refer to the generators in $\Sigma_{\textbf{NBiM}}$.
\end{enumerate}
\end{example}

\subsubsection{Rewriting in a PROP.}

\begin{notation}
Note that we write generic pairs and tuples using parentheses and reserve the notation 
$\langle l, r \rangle$ specifically for the case when the pair $(l, r)$ forms a rewriting rule.
\end{notation}

\begin{definition} \label{defn:rewprop}
A \emph{rewriting system} $\RS$ in a PROP $\catA$ consists of a set of \emph{rewriting rules}, i.e. pairs $\rrule{l}{r}$ of arrows $l, r \: i \to j$ in $\catA$ with the same arities and coarities. Given $a,b \: m \to n$ in $\catA$, $a$ rewrites into $b$ via $\RS$, written $a \Rew{\RS} b$, if they are decomposable as follows, for some rule $\rrule{l}{r} \in \RS$
\begin{equation}
\label{eq:rewpropmatch}
\input{./figures/lhs-ctx.tikz} \qquad\qquad
\input{./figures/rhs-ctx.tikz}
\end{equation}
In this situation, we say that $a$ contains a \emph{redex} for $\rrule{l}{r}$.
\end{definition}

The following well-known example illustrates the subtlety of critical pair analysis when rewriting in monoidal categories.
\begin{example}[From \cite{Lafont2003}, see also \cite{Mimram14}]\label{ex:mimram}
Fix $\Sigma = \{ \gamma \colon 2 \to 2 \}$ and consider the rewriting system on $\syntax{\Sigma}$ consisting of the following rule
\begin{equation}\label{eq:mimramYB-prop}
\cgr[height=1cm]{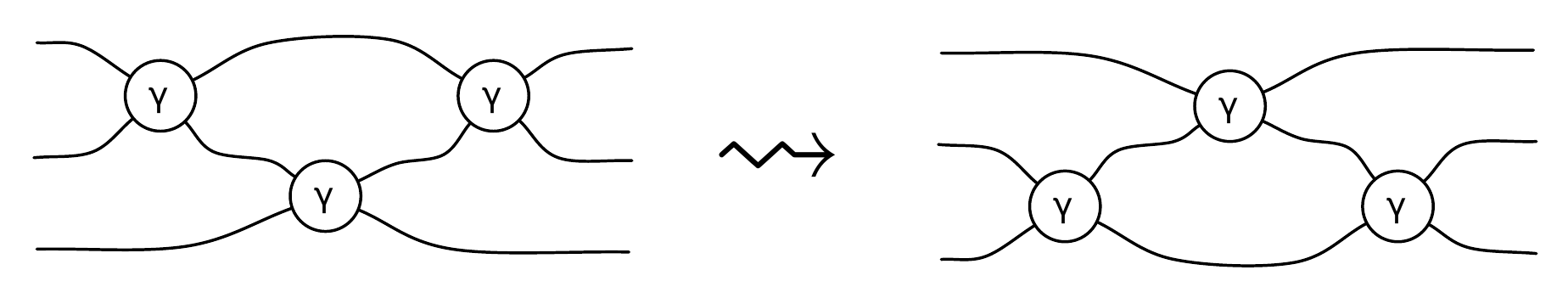}
\end{equation}
A critical pair analysis yields an infinite number of critical pairs.
Indeed, as shown in~\cite{Lafont2003,Mimram14}, any diagram $\phi:1+m \to 1+n$ that does not decompose non-trivially into
$\phi=\mu + \nu$ for some $\mu,\nu$ yields a critical pair

\[
\includegraphics[height=2.5cm]{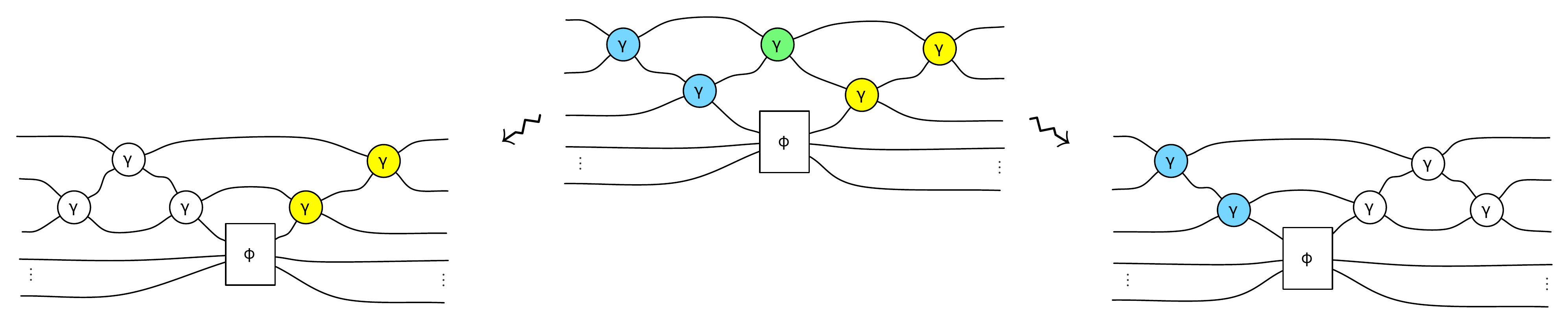}
\]
in which clearly there are two embeddings of the left-hand side\ of~\eqref{eq:mimramYB-prop} (depicted in blue and yellow, respectively,
in a colour version of the paper) with an overlap (in green).
\end{example}

In \cite{MimramFix} this problem was solved by adding duals to monoidal categories.
In Section~\ref{sec:confluenceFrob}, we will show another solution based on~\cite{BGKSZ-parttwo}:
a translation from PROPs to DPO rewriting with interfaces. The example below anticipates this encoding. It will be useful as a running example for the next section, which is devoted to critical pair analysis and confluence in DPO rewriting with interfaces.


\begin{example}\label{ex:mimramdpo}
Treating the rewriting system of Example~\ref{ex:mimram} as DPO system over $\Hyp{\Sigma}$
with $\gamma:2\to 2\in \Sigma$ yields the following DPO rule
\[ \scalebox{0.7}{\input{./figures/yang-baxter.tikz}} \] 
The formal correspondence between PROPs and DPO rewriting with interfaces will be explained in Section~\ref{sec:confluenceFrob}. For the time being, the reader can observe the similarities between the left-hand side of~\eqref{eq:mimramYB-prop} and the left-hand side of the above rule:  each $\gamma$ in ~\eqref{eq:mimramYB-prop} corresponds to an hyperedge (labeled with $\gamma$) in the hypergraph above; moreover, each wires in~\eqref{eq:mimramYB-prop} corresponds to a node above; finally, dangling wires in~\eqref{eq:mimramYB-prop} are exactly the numbered nodes. A similar correspondence holds for the right-hand side, while the interface of the DPO rule, depicted in light blue, just collects all the numbered nodes.  
 
Below, we give a DPO derivation with interface (in light blue), corresponding to a critical pair from the family identified in Example~\ref{ex:mimram}
\[
  \scalebox{0.4}{\input{./figures/big-critical-pair.tikz}}
\]
\end{example}




\section{Confluence for DPO rewriting with interfaces}\label{sec:confluence}

Differently from Definition~\ref{defn:DPOcritpair}, the interface of the pre-critical pair plays a crucial role when considering the setting of DPO with interfaces.

\begin{definition}[Pre-critical pair with interface]\label{pre-critical} Let $\mathcal{R}$ be a DPO system with rules $L_1 \tl{} K_1 \tr{} R_1$ and $L_2 \tl{} K_2 \tr{} R_2$. Consider two derivations with source $S\leftarrow J$
\begin{equation}\label{diag:criticalpairInterface}
\vcenter{
    \xymatrix@R=12pt{
    R_1 \ar[d] & \ar[l] K_1 \dlcorner  \ar[d] \ar[r] & L_1 \ar@{}[dr]|(.8){\text{\large $\ulcorner$}\qquad\quad} \ar[dr]^{f_1} && \ar[dl]_{f_2} \ar@{}[dl]|(.8){\qquad\quad\text{\large $\urcorner$}} L_2 & \ar[l] K_2 \ar[d] \drcorner\ar[r] & R_2 \ar[d] \\
    H_1  & \ar[l] C_1 \ar[rr] & & S  & & \ar[ll] C_2  \ar[r] & H_2 \\
    &&& \ar@/^/[ull] J \ar@{}[u]|{(\dagger)} \ar@/_/[urr] & & &
    }
}
\end{equation}
We say that $ (H_1\tl{}J) \LDPOstep{}  (S \tl{}J ) \DPOstep{} (H_2\tl{}J)$ is a \emph{pre-critical pair} if $[f_1,f_2] \colon L_1 + L_2 \to S$ is epi
and $(\dagger)$ is a pullback;
it is \emph{joinable} if there exists $W\tl{}J$ such that
$(H_1\tl{}J) \DPOstep{}^* (W\tl{}J) {\,\,{}^*\!\!\LDPOstep{} \,\,\,} (H_2\tl{}J)$.
\end{definition}

Definition~\ref{pre-critical} augments Definition~\ref{defn:DPOcritpair} with the interface $J$, given by ``intersecting'' $C_1$ and $C_2$. Intuitively, $J$ is the largest interface that allows both rewriting steps.

\begin{example}\label{exm:plumpinterface}
Consider the pair of rewriting steps \eqref{eq:criticalpairplumprevisited} in Example \ref{ex:againstPlump}. This is a pre-critical pair: the reader can check that the interface is indeed a pullback, constructed as in $(\dagger)$. Observe moreover that this pair is \emph{not} joinable.
Should we consider rewriting without interfaces, i.e., should $J$ be the empty graph, 
the pair 
$(H_1\tl{}J) \LDPOstep{}  (S \tl{}J ) \DPOstep{} (H_2\tl{}J)$
would not be a pre-critical pair anymore.
\end{example}

Plump's Example~\ref{ex:Plump} shows that in ordinary DPO, joinability of pre-critical pairs does not imply confluence. Our Example~\ref{exm:plumpinterface} shows that the argument does not work for DPO with interfaces. Indeed,  as we shall see in Theorem \ref{th:locconfl}, in the presence of interfaces  joinability suffices for confluence. To prove it, we assume the following property.

\begin{assumption}\label{assumption}
Our ambient category $\catC$ is assumed (1) to possess an epi-mono factorisation system, (2) to have binary coproducts, pushouts and pullbacks, and (3) to be adhesive with (4) all the pushouts stable under pullbacks.
\end{assumption}

The above conditions hold in any presheaf category. Additionally, they are closed under slice.
It follows that $\Hyp{\Sigma}$ is an example of such a category.

\medskip

We could now mimic the definition of parallel pairs given in Definition~\ref{defn:DPOparpair}.
However, the existence of pullbacks in property (2) allows for a simpler characterisation.

\begin{definition}[Parallel pair with interface]
\label{parallelpairs}
Let $\mathcal{R}$ be a DPO system with rules $L_1 \tl{} K_1 \tr{} R_1$ and $L_2 \tl{} K_2 \tr{} R_2$. Consider two derivations with common source $S\leftarrow J$
as in Definition~\ref{pre-critical} and the diagram below formed by pullbacks

\begin{equation*}
\xymatrix@R=10pt{
   & L_1 \ar[rd] & & L_2 \ar[ld] \\
   X \ar[rd] \ar[ru] & & S & & Y \ar[ld] \ar[lu] \\
   & C_2 \ar[ru] & & C_1 \ar[lu] \\
   & & J \ar[ru] \ar[lu] \\
      }
\end{equation*}

\noindent
We say that $ (H_1\tl{}J) \LDPOstep{} (S \tl{}J ) \DPOstep{} (H_2\tl{}J)$ is a \emph{parallel pair} if $X \to L_1$ and $Y\to L_2$ are iso and $C_1 \to S$ and $C_2\to S$ are mono.
\end{definition}

The definition is slightly stronger than the one for parallel independence for DPO rewriting without
interfaces shown in Definition~\ref{defn:DPOparpair}, even if it coincides whenever rules are
left-linear~\cite[Definition~5 and Proposition~1]{CorradiniDLRMCA18}.
However, this formulation is better suited for our notion of rewriting \emph{with} interfaces, and indeed,
it is easy to see that parallel pairs are joinable.

Before moving to the proof, though,
we need a technical lemma, the following simple pushout decomposition result
(aka ``mixed decomposition'' from \cite{BaldanGS11}). The proof uses only
stability of pushouts under pullbacks, which is encompassed by our Assumption~\ref{assumption}.

\begin{lemma}\label{lem:decomp}
Suppose that in the diagram below $m$ is mono, $(\dagger)+(\ddagger)$ is a pushout, and $(\ddagger)$
is a pullback. Then both $(\dagger)$ and $(\ddagger)$ are pushouts.
\begin{equation}\label{eq:pushpull}
\raise15pt\hbox{$
\xymatrix@R=10pt@C=25pt{
K \ar@{}[dr]|{(\dagger)}
\ar[d] \ar[r] & {C'}
\ar@{}[dr]|{(\ddagger)}
\ar[d] \ar[r] & C \ar[d] \\
L \ar[r] & {G'} \ar[r]_m & G
}$}
\end{equation}
\end{lemma}

\begin{proposition}
\label{confparallelpairs}
Let $\mathcal{R}$ be a DPO system with rules $L_1 \tl{} K_1 \tr{} R_1$ and $L_2 \tl{} K_2 \tr{} R_2$.
Consider two derivations with common source $S\leftarrow J$ as in Definition~\ref{pre-critical}.
If $ (H_1\tl{}J) \LDPOstep{} (S \tl{}J ) \DPOstep{} (H_2\tl{}J)$ is a parallel pair then it is joinable.
\end{proposition}

\begin{proof}
Let us assume that $X = L_1$ and $Y = L_2$, so that the arrows
$L_1 \to C_2$ and $L_2 \to C_1$ are obviously defined, and the conditions of
Definition~\ref{defn:DPOparpair} satisfied.
The existence of a parallel pair for graph without interfaces is a standard result,
see for example the survey~\cite{HabelMP01}.

Note also that $C_i \to S$ are both mono, even if $K_i \to L_i$ might not be so.
Hence, we can retrace the steps of the classical proof presented in~\cite[Section 9.7]{Ehrig79}.
So, consider the two diagrams below

\begin{equation}\label{eq:pardec1}
\raise15pt\hbox{$
\xymatrix@R=10pt@C=25pt{
 & {K_2}
\ar@{}[dr]|{(3)}
\ar[d] \ar[r] & L_2 \ar[d] \\
K_1 \ar@{}[dr]|{(2)}
\ar[d] \ar[r] & {J}
\ar@{}[dr]|{(1)}
\ar[d] \ar[r] & C_1 \ar[d] \\
L_1\ar[r] & {C_2} \ar[r] & S
}$}
\hspace{1cm}
\raise15pt\hbox{$
\xymatrix@R=10pt@C=25pt{
 & {K_2}
\ar@{}[dr]|{(6)}
\ar[d] \ar[r] & R_2 \ar[d] \\
K_1 \ar@{}[dr]|{(5)}
\ar[d] \ar[r] & {J}
\ar@{}[dr]|{(4)}
\ar[d] \ar[r] & E_2 \ar[d] \\
R_1\ar[r] & {E_1} \ar[r] & W
}$}
\end{equation}

The square (1) is a pullback, hence $K_i \to J$ are obtained by the universal property.
By construction the squares (2)+(1) and (3)+(1) are pushouts, hence
by the mixed decomposition lemma the squares (1), (2), and (3) are also pushouts.
Squares (4), (5), and (6) are now obtained by pushout.

Consider then the diagrams below

\begin{equation}\label{eq:pardec2}
\raise15pt\hbox{$
\xymatrix@R=10pt@C=25pt{
 {K_1}
\ar@{}[dr]|{(5)}
\ar[d] \ar[r] & R_1 \ar[d] \\
{J}
\ar@{}[dr]|{(7)}
\ar[d] \ar[r] & E_1 \ar[d] \\
{C_1} \ar[r] & H_1
}$}
\hspace{1cm}
\raise15pt\hbox{$
\xymatrix@R=10pt@C=25pt{
 {K_2}
\ar@{}[dr]|{(6)}
\ar[d] \ar[r] & R_2 \ar[d] \\
{J}
\ar@{}[dr]|{(8)}
\ar[d] \ar[r] & E_2 \ar[d] \\
{C_2} \ar[r] & H_2
}$}
\hspace{1cm}
\raise15pt\hbox{$
\xymatrix@R=10pt@C=25pt{
L_2  \ar[d]& {K_2} \ar@{}[dl]|{(3)}
\ar@{}[dr]|{(6)}
\ar[l] \ar[d] \ar[r] & R_2 \ar[d] \\
C_1 \ar@{}[dr]|{(7)}
\ar[d]  & {J} \ar[l]
\ar@{}[dr]|{(4)}
\ar[d] \ar[r] & E_2 \ar[d] \\
H_1 & {E_1} \ar[l] \ar[r] & W
}$}
\end{equation}

Squares (5), (6), (5)+(7), and (6)+(8) are pushouts, hence $E_i \to H_i$ are obtained by the universal property and
also squares (7) and (8) are pushouts.

We then have all in place to obtain two derivations
$(H_1\tl{}J) \DPOstep{} (W\tl{}J) {\,\,{}\!\!\LDPOstep{} \,\,\,} (H_2\tl{}J)$: The
derivation $(H_1\tl{}J) \DPOstep{} (W\tl{}J)$ is depicted on the right
in the picture above.
\end{proof}

As for the rewriting without interfaces, in our results we stick to pre-critical pairs, the distinction being immaterial, even if
in the examples we usually show just the critical ones.

The following construction mimics \cite{EhrigHPP04}. It allows us to
restrict --or ``clip''-- a DPO rewriting step with match $f \colon L\to G$ to any subobject
of $G'$ through which $f$ factors.
\begin{construction}[One-step clipping] \label{constr:clipping}
Suppose we have a DPO rewriting step as below left, 
together with factorisation $L\tr{} G'\xrightarrow{m} G$ where $m$ is mono. As shown below right, we get $C'$ by pulling back $G'\tr{}G \tl{} C $ and $K \to C'$ by the universal property
$$\xymatrix@R=5pt@C=10pt{
& L \ar[dd] \ar[ld] & & K \ar[rr] \ar[ll]  \ar[dd] & & R \ar[dd] \ \\
G' \ar[rd]_m & &   & &  \\
& G & & C \ar[rr] \ar[ll] & & H \\
}
\qquad \qquad
\xymatrix@R=5pt@C=10pt{
& L \ar[dd] \ar[ld] & & K \ar[rr] \ar[ll] \ar@{.>}[ld] \ar[dd] & & R \ar[dd]  \\
G' \ar[rd]_m & & C' \ar[ll] \ar[rd] & & \\
& G & & C \ar[rr] \ar[ll] & & H \\
}$$
By Lemma~\ref{lem:decomp}
the two leftmost squares are both pushouts.
Next, $H'$ is the pushout of $C' \tl{}  K \tr{} R$ and $H'\tr{} H$ follows from its universal property
$$\xymatrix@R=5pt@C=10pt{
& L \ar[dd] \ar[ld] & & K \ar[rr] \ar[ll] \ar@{.>}[ld] \ar[dd] & & R \ar[dd]  \ar[ld]\\
G' \ar[rd]_m & & C' \ar[ll] \ar [rr] \ar[rd] & & H' \ar@{.>}[rd]\\
& G & & C \ar[rr] \ar[ll] & & H \\
}$$
By pushout pasting also the bottom-rightmost square is a pushout.
Finally, observe that $C'\to C$ is mono since it is the pullback of $m$
along $C\to G$. This means that each of the two squares in diagram below is, as well as being a pushout, also a pullback, since each is a pushout along a mono in an adhesive category
\[
\xymatrix@R=15pt@C=25pt{
G' \ar[d]_m & C' \ar[d] \ar[l] \ar[r] & H' \ar[d] \\
G & C \ar[l] \ar[r] & H
}
\]
 \end{construction}

\begin{example}\label{ex:clipping}
We use the clipping construction to restrict
 pairs of derivations with common source into pre-critical pairs.
For example, consider the
two DPO rewriting rules illustrated in Example~\ref{ex:mimramdpo}.
We can factorise the two matches through their common image, and clip, as illustrated below
\[
  \scalebox{0.45}{\input{./figures/big-critical-pair-clipping.tikz}}
\]
\end{example}
Note that the clipped derivations result with the two matches being jointly epi,
which is one of the properties of a pre-critical pair. This generalises:
given two rewriting steps with common source
 $(G_{1,1} \tl{} I)  \LDPOstep{} (G_0 \tl{} I) \DPOstep{} (G_{1,2} \tl{} I)$,
the next construction produces a pre-critical pair
$(G_{1,1}'\tl{} J') \LDPOstep{} ({G_0}' \tl{} J') \DPOstep{} (G_{1,2}'\tl{} J')$
using clipping.
\begin{construction}[Pre-critical pair extraction]\label{constr:extraction}
Start with two rewrites from $G_0\tl{} I$
$$\xymatrix@R=12pt{
    R_{1,1} \ar[d] & \ar[l] K_{1,1} \dlcorner  \ar[d] \ar[r] & L_{1,1} \ar@{}[dr]|(.8){\text{\large $\ulcorner$}\qquad\quad} \ar[dr]^{f_1} && \ar[dl]_{f_2} \ar@{}[dl]|(.8){\qquad\quad\text{\large $\urcorner$}} L_{1,2} & \ar[l] K_{1,2} \ar[d] \drcorner\ar[r] & R_{1,2} \ar[d] \\
    G_{1,1}  & \ar[l] C_{1,1} \ar[rr] & & G_0  & & \ar[ll] C_{1,2}  \ar[r] & G_{1,2} \\
    &&& \ar@/^/[ull] I \ar[u] \ar@/_/[urr] & & &
    }$$
and factorise $[f_1,f_2]\colon L_{1,1}+L_{1,2} \to G_0$ to obtain
  $$ \vcenter{
    \xymatrix@R=12pt@C=5pt{
    L_{1,1} \ar[rd] \ar@(d,ul)[rdd]_{f_1} & & L_{1,2} \ar[ld] \ar@(d,ur)[ldd]^{f_2} \\
    & G_0' \ar@{_{(}->}[d] \\
    & G_0.
    }
}$$

Next apply Construction~\ref{constr:clipping} twice, obtaining

$$ \scalebox{.8}{    \xymatrix@R=10pt@C=10pt{
  &  R_{1,1} \ar[dd]|\hole \ar[dl] & &  \ar[ll] K_{1,1} \ar[dl]  \ar[dd]|\hole \ar[rr] & &  L_{1,1} \ar[drr] \ar@/_1pc/[ddrr]|(.36)\hole && && \ar@/^1pc/[ddll]|(.36)\hole  L_{1,2} \ar[dll] & &  \ar[ll] K_{1,2}  \ar[dr] \ar[rr] \ar[dd]|\hole && R_{1,2} \ar[dr] \ar[dd]|\hole \\
  H_{1,1}' \ar[dr] && C_{1,1}' \ar[dr] \ar[ll]  \ar[rrrrr] & & & & & G_0' \ar[d] & & &&& C_{1,2}' \ar[dl]  \ar[lllll] \ar[rr] && H_{1,2}'   \ar[dl]\\
  &  H_{1,1}  & & \ar[ll] C_{1,1} \ar[rrrr] && & & G_0 && & & \ar[llll] C_{1,2}  \ar[rr] & & H_{1,2}. \\
  & &&& &&& \ar@/^/[ullll] I \ar[u] \ar@/_/[urrrr] & & &
    }
}$$

Finally, pull back $C_{1,1}' \tr{} G_0' \tl{} C_{1,2}'$ to obtain the pre-critical pair

 $$ \scalebox{.8}{
    \xymatrix@R=8pt@C=10pt{
  &  R_{1,1}  \ar[dl] & &  \ar[ll] K_{1,1} \ar[dl]  \ar[rr] & &  L_{1,1} \ar[drr]  && &&   L_{1,2} \ar[dll] & &  \ar[ll] K_{1,2}  \ar[dr] \ar[rr]  && R_{1,2} \ar[dr] \\
  G_{1,1}'  && C_{1,1}'  \ar[ll]  \ar[rrrrr] & & & & & G_0'  & & &&& C_{1,2}'  \ar[lllll] \ar[rr] && G_{1,2}'.   \\
  &   & &  && & & J'  \ar@/^5pt/[ulllll] \ar@/_5pt/[urrrrr] && & &    }
}$$
\end{construction}

\begin{example}\label{ex:extraction}
We can now complete the pre-critical pair extraction process, commenced in Example~\ref{ex:clipping},
following the steps of Construction~\ref{constr:extraction}
\[
  \scalebox{0.4}{\input{./figures/big-critical-pair-clipped.tikz}}
\]
\end{example}

Construction~\ref{constr:extraction} means that we are able to extract a pre-critical pair from two rewriting steps with common source. If the pre-critical pair is joinable, we would then like to embed the joining derivations to the original context.

The following is a useful step in this direction. Assuming a mono $G_0'\to G_0$, it allows us
to extend a derivation from $G_0' \tl{} J'$ to a corresponding one from $G_0 \tl{} J$, if
we can obtain $G_0$ by glueing $G_0'$ and some context $C_0$ along $J'$. Stated more formally, we want the following diagram commute and $(\dagger)$ be a pushout
\begin{equation}\label{diag:embeddingAssumption}
\begin{aligned}
\xymatrix@R=10pt{
& J' \ar[r] \ar[d] \ar@{}[dr]|(.5){(\dagger)} & G_0' \ar[d] \\
J  \ar[r]  \ar@/_10pt/[rr]
& C_0  \ar[r] & G_0
}
\end{aligned}
\end{equation}

 \begin{construction}[Embedding]\label{constr:embedding}
The extended derivation is constructed as in the commuting diagram below, where each square is a pushout diagram

\begin{equation*}\label{diag:embedding}
\begin{aligned}
      \scalebox{.7}{
     \xymatrix@R=20pt@C=7pt{
       & & & L_1 \ar[ld] & K_1 \ar[d]\ar[l] \ar[r] & R_1\ar[rd] & & L_2 \ar[ld] & K_2 \ar[d]\ar[l] \ar[r] & R_2 \ar[rd] & & \dots & & L_n \ar[ld] & K_n \ar[d] \ar[l] \ar[r] & R_n \ar[rd] \\
     J' 
        \ar@{}[rrd]|(0.5){(\dagger)} \ar[rr] \ar[d]
       & & G_0' \ar[d] & & C_1' \ar@{}[dll]|(0.50){(\gamma_1)} \ar@{}[drr]|(0.50){(\delta_1)}  \ar[d] \ar[ll] \ar[rr] & & G_1' \ar[d]
       & & C_2'  \ar@{}[dll]|(0.50){(\gamma_2)} \ar@{}[drr]|(0.50){(\delta_2)}  \ar[d] \ar[ll] \ar[rr] && G_2' \ar[d] & \dots &  G_n' \ar[d]
       & & C_n' \ar@{}[dll]|(0.50){(\gamma_n)} \ar@{}[drr]|(0.50){(\delta_n)}  \ar[d] \ar[ll] \ar[rr] && G_n' \ar[d] \\
     C_0 \ar[rr] 
       & & G_0 & & C_1 \ar[ll] \ar[rr] & & G_1 & & C_2 \ar[ll] \ar[rr] & & G_2 & \dots &  G_n & & C_n \ar[ll] \ar[rr] & & G_n \\
       & & & & & & & & & & &  \\
       & & & & & & & & & &
       & J' 
           \ar[d] \ar[uuulll]|(0.3){} \ar@(dl,d) [uuulllllll]| (0.3) {} \ar[uuurrr]|(0.3){} \ar@{}[uulll]|(0.5){(\epsilon_2)} \ar@{}[uurrr]|(0.5){(\epsilon_n)} &
       & & & & & \\
       & & & & & & & & & &
       & C_0 \ar[uuulll]|(0.3){} \ar@(dl,d)[uuulllllll]|(0.56){} \ar[uuurrr]|(0.3){}  \ar@{}[uuullllllll]|(0.5){(\epsilon_1)} & & J \ar[ll] & & & &
      }
    }
\end{aligned}
\end{equation*}

We shall now explain each of the components. The upper row of pushouts together with morphisms $J'\tr{}C_i'$ witnesses the original derivation $(G_0'\tl{}J') \DPOstep{}^* (G'_n \tl{}J')$.

For $i=1\dots n$, $(\epsilon_i)$ is formed as the pushout of
$C_0 \tl{} J'  \tr{ } C'_i$ and $(\delta_i)$ as the pushout of $ C_i \tl{} C_i'  \tr{ }G'_i$, as shown in the diagram below 

\begin{equation} \label{diag:embeddingEi-Di}
\begin{aligned}
     \xymatrix@R=15pt@C=20pt{
     J' \ar[r] \ar[d] &C_i' \ar[d] \ar[r]& G_i'  \ar[d] \\
     C_0 \ar[r]  \ar@{}[ru]|{(\epsilon_i)} & C_i \ar[r] \ar@{}[ru]|{(\delta_i)}  & G_i
     }
\end{aligned}
\end{equation}

It remains to construct pushouts $(\gamma_i)$, which is done in the following diagram
\begin{equation}
\label{diag:embeddingCi}
\begin{aligned}
\scalebox{1}{
     \xymatrix@R=4pt@C=20pt{
      J'   \ar[rr] \ar[dr] \ar[ddd] && G_{i-1}'  \ar[ddd]\\
       &C_i' \ar[d] \ar[ur]& \\
       & C_i \ar@{.>}[dr]& \\
      C_0 \ar[ur] \ar[rr] \ar@{}[ruuu]|{(\epsilon_i)} && G_{i-1}  \ar@{}[luuu]|{(\gamma_i)}
     }}
\end{aligned}
\end{equation}

The exterior square in~\eqref{diag:embeddingCi} is a pushout: for $i=1$ it is $(\dagger)$ from~\eqref{diag:embeddingAssumption}, while for $i\geq 2$ it is obtained by composing $(\epsilon_{i-1})$ and $(\delta_{i-1})$ from~\eqref{diag:embeddingEi-Di}. The universal property of $(\epsilon_i)$ yields the morphism $C_i \to G_{i-1}$. By pushout decomposition, the diagram $(\gamma_i)$ is a pushout.

 \end{construction}

\begin{example}
In Example~\ref{ex:extraction} we saw two derivations from
\[
  \scalebox{0.75}{\input{./figures/big-critical-pair-graph.tikz}}
\]
These can be extended to
\[
  \scalebox{0.75}{\input{./figures/big-critical-pair-ext.tikz}}
\]
following the steps in Construction~\ref{constr:embedding} because the
square below is a pushout
\[
  \scalebox{0.5}{\input{./figures/big-critical-pair-ext-pushout.tikz}}
\]
\end{example}

Constructions~\ref{constr:extraction} and~\ref{constr:embedding} are the main ingredients for
showing the Knuth-Bendix property for DPOI. Before we prove it, we need
one technical lemma from the theory of adhesive categories.

\begin{lemma}\label{lemmacube}
Consider the cube below, where the top and bottom faces are pullbacks,
the rear faces are both pullbacks and pushouts, and $m$ is mono. Then, the front faces are also pushouts.
\end{lemma}
$$\xymatrix@R=3pt@C=5pt{
& G_0' \ar[dd]^(.3)m|\hole & & C_{1,2}' \ar[ll] \ar[dd] \\
C_{1,1}' \ar[ur] \ar[dd]_n && J' \ar[dd] \ar[ur] \ar[ll]\\
& G_0 & & C_{1,2} \ar[ll]|\hole \\
C_{1,1} \ar[ur] && J \ar[ur] \ar[ll]
}
$$

 \begin{theorem}[Local confluence] \label{th:locconfl} For a DPO system with interfaces, if all pre-critical pairs are joinable then rewriting is locally confluent: given
 $(G_{1,1} \tl{} I) \LDPOstep{} (G_0 \tl{} I) \DPOstep{} (G_{1,2} \tl{} I)$, there exists $W \tl{} I$ such that
 \[
\xymatrix@R=5pt@C=10pt{
&  \ar@{~>}[dl]  (G_0 \tl{} I) \ar@{~>}[dr] & \\
(G_{1,1} \tl{} I) \ar@{~>}[dr]^>>>{*} && (G_{1,2} \tl{} I) \ar@{~>}[dl]_>>>{*} \\
&(W \tl{} I).&
}
\]
 \end{theorem}

\begin{proof}
 Following the steps of Construction~\ref{constr:extraction},
 we obtain a pre-critical pair
 \[
 (G_{1,1}'\leftarrow J') \LDPOstep{} (G_0'\leftarrow J')\DPOstep{} (G_{1,2}' \leftarrow J')
 \]
Because pre-critical pairs are by assumption joinable we have derivations $$(G_{1,1}' \tl{} J') \DPOstep{}^* (W' \tl{\beta'} J' ) \DPOtlrewr{} (G_{1,2}' \tl{} J') \text{.}$$
Suppose that the leftmost derivation requires $n$ steps and the rightmost $m$. To keep the notation consistent with Construction~\ref{constr:embedding}, we fix notation $G_{n,1}' \df W' {=:} \ G_{m,2}'$.

Now let $J$ be the pullback object of $C_{1,1}\tr{} G_0 \tl{}C_{1,2}$. By the universal property, we obtain maps $\iota\colon I\to J$ and $\xi\colon J'\tr{} J$
\begin{equation}\label{diag:cubemainproof}
\begin{aligned}
\xymatrix@R=5pt@C=10pt{
& G_0' \ar[dd]|\hole & & C_{1,2}' \ar[ll] \ar[dd] &&& \\
C_{1,1}' \ar[ur] \ar[dd] && J' \ar@{.>}[dd]^(0.3)\xi \ar[ur] \ar[ll] &&& \\
& G_0 & & C_{1,2} \ar[ll]|(.53)\hole &&& \\
C_{1,1} \ar[ur] && J \ar[ur] \ar[ll] &&& I \ar@{.>}[lll]_\iota \ar@/_1pc/[ull] \ar@/^1pc/[lllll]
}
\end{aligned}
\end{equation}

\smallskip
Recall by Construction~\ref{constr:clipping} that the rear faces of \eqref{diag:cubemainproof} are both pullbacks and pushouts. Then, by Lemma \ref{lemmacube}, the square below is a pushout
$$
\xymatrix@R=10pt@C=10pt{J' \ar[d] \ar[r]  \ar@{}[dr]|(0.5){(\dagger)} & G_0' \ar[d] \\
J \ar[r] & G_0
}
$$

We are now in position to apply Construction~\ref{constr:embedding} by taking $C_0=J$, which yields
$$(G_0 \tl{} J) \DPOstep{} (G_{1,1} \tl{} J) \DPOstep{}^* (G_{n,1} \tl{\beta_1} J )$$
 extending $(G_0' \tl{} J') \DPOstep{} (G_{1,1}' \tl{} J') \DPOstep{}^* (G_{n,1}' \tl{\beta'} J' )$
   and
  $$(G_0 \tl{} J) \DPOstep{} (G_{1,2} \tl{} J) \DPOstep{}^* (G_{m,2} \tl{\beta_2} J )$$
   extending $(G_0' \tl{} J') \DPOstep{} (G_{1,1}' \tl{} J') \DPOstep{}^* (G_{m,2}' \tl{\beta'} J' )$.

 The next step is to prove that $(G_{n,1} \tl{\beta_1} J ) \cong (G_{m,2} \tl{\beta_2} J ) $. To see this, it is enough to observe that both the following  squares are pushouts of $J \tl{\xi} J' \tr{\beta'} W'=G_{n,1}'=G_{m,2}'$
  $$
 \xymatrix@=10pt{
 J' \ar[d]_\xi \ar[r]^{\beta'} & G_{n,1}' \ar[d] \\
 J \ar[r]_{\beta_1} & G_{n,1}
 }
 \qquad
 \xymatrix@=10pt{
 J' \ar[d]_\xi \ar[r]^{\beta'} & G_{m,2}' \ar[d] \\
 J \ar[r]_{\beta_2} & G_{m,2}
 }
$$
Indeed, the leftmost is a pushout by composition of squares $(\epsilon_n)$ and $(\delta_n)$ in the embedding construction and the rightmost by composition of $(\epsilon_m)$ and $(\delta_m)$.

 To complete the proof, it remains to show that, in the above derivations, interface $J$ extends to interface $I$ as in the statement of the theorem. But this trivially holds by precomposing with
 $\iota\colon I\to J$.
\end{proof}

We are now ready to give our decidability result. To formulate it at the level of generality of adhesive categories we need some additional definitions.

A \emph{quotient}
of an object $X$ is an equivalence class of epis with domain $X$. Two epis $e_1: X\ra X_1$,
$e_2\colon X\ra X_2$ are equivalent when there exists an isomorphism $\varphi\colon X_1\to X_2$ such that
$e_1 \poi \varphi  = e_2$.
Note that quotient is the dual of \emph{subobject}.

A DPO rewriting system with interfaces
is \emph{computable} when
\begin{itemize}
\item pullbacks are computable,
\item for every pair of rules $L_i\tl{}K_i \to R_i$, $L_j\tl{}K_j \to R_j$, the set of quotients of $L_i+L_j$ is finite and computable,
\item for all $G\tl{} J$, it is possible to compute every $H\tl{}J$ such that $(G\tl{} J) \DPOstep{} (H\tl{}J)$.
\end{itemize}

Computability refers to the possibility of effectively computing each rewriting step as well as to have a finite number of pre-critical pairs.
More precisely, the first two conditions ensure that the set of all pre-critical pairs is finite (since every object has finitely many quotients)
and each of them can be computed,
while the last one ensures that any possible rewriting step can also be computed.
Thus, these assumptions rule out the rewriting of infinite structures, singleing out instead those structures where it is reasonable to
apply the DPO mechanism,
like finite hypergraphs in $\Hyp{\Sigma}$, which are exactly what is needed for implementing rewriting of SMTs.

\begin{corollary}\label{cor:comput}
For a computable terminating DPO system with interfaces, confluence is decidable.
\end{corollary}
\begin{proof}
We first observe that an arbitrary DPO system with interface is confluent if and only if all pre-critical pairs are joinable.
\begin{enumerate}
\item If all pre-critical pairs are joinable then, by Theorem \ref{th:locconfl},  the system is confluent.
\item  If not all pre-critical pairs are joinable, then at least one pair witnesses the fact that the system is not confluent. 
\end{enumerate}
Therefore, to decide confluence, it suffices to check that all pre-critical pairs are joinable.

\medskip

Since the system is computable, there are only finitely many pre-critical pairs and these can be computed. For each pair, one can decide joinability: Indeed, each rewriting step can be computed (since the system is computable) and there are only finitely many $(H\tl{}J)$ such that $(G \tl{}J) \DPOstep{}^* (H\tl{}J)$ (since the system is terminating).
\end{proof}

It is worth to remark that this result is not in conflict with Theorem \ref{thm:plump}: Corollary~\ref{cor:comput} refers to the confluence of all hypergraphs with interfaces $G\tl{}J$. 
The property that Theorem \ref{thm:plump} states as undecidable is whether the rewriting is confluent for all \emph{hypergraphs with empty interface} $G\tl{}0$. Observe that 
the restriction to hypergraphs with empty interface would make the above proof fail in point (2): indeed, thanks to Theorem~\ref{th:locconfl}, point (1) would hold also for 
hypergraphs with empty interface, but a non-joinable pre-critical pair originating from $S\tl{}J$ with $J$ non empty does not necessarily witness that rewriting is not confluent 
for all $G\tl{}0$.

A similar problem arises with term rewriting, when restricting to the confluence of \emph{ground terms} \cite{kapur1990ground}. As an example, consider the following term rewriting system defined on the signature with two unary symbols, $f$ and $g$, and one constant $c$
$$f(g(f(x))) \to x \quad \quad \quad f(c)\to c \quad \quad \quad g(c)\to c$$
The critical pair $f(g(x)) \leftarrow f(g(f(g(f(x))))) \to g(f(x))$ is not joinable, but the system is obviously ground confluent, as every ground term will eventually rewrite into $c$.

Our work therefore allows one to view Theorem \ref{thm:plump} in a new light: as hypergraphs with empty interface are morally the graphical analogous of ground terms, we can say that ground confluence is not decidable for DPO rewriting with interfaces.


\section{Confluence for PROP rewriting with Frobenius structure}\label{sec:confluenceFrob}

As emphasised in the introduction, a major reason for interest in DPO rewriting with interfaces is that PROP rewriting (\S\ref{ssec:PROPrewriting}) may be interpreted therein. In this section we investigate how our confluence result behaves with respect to this interpretation. The outcome is that confluence is decidable for terminating PROP rewriting systems, where terms are taken modulo a chosen special Frobenius structure (Corollary~\ref{cor:confluenceFROB}). For arbitrary symmetric monoidal theories, confluence is also decidable, provided that certain additional conditions hold (Corollary~\ref{cor:confluenceSMTs}).

\subsection{From PROPs to hypergraphs with interfaces}\label{ssec:PROPDPOrewriting}

In this subsection we report a result from~\cite{BGKSZ-partone}
that is crucial for the encoding of PROP rewriting into DPO rewriting with interfaces in $\Hyp{\Sigma}$ (cf. Section \ref{ssec:DPOrewriting}).



First, we obtain our domain of interpretation by restricting the category $\Cospan{\Hyp{\Sigma}}$ whose objects are  hypergraphs 
and arrows cospans, i.e. pairs $G_1 \tr{} G_2 \tl{}G_3$ of $\Sigma$-hypergraphs morphisms, up-to isomorphism in the choice of $G_2$.

\begin{definition}[Hypergraphs with interfaces] \label{def:frobtermgraph}
Any $k\in \N$ can be seen as a discrete hypergraph (i.e., with an empty set of edges) with $k$ vertices.
The objects of the PROP
 $\FTerm{\Sigma}$ 
 are natural numbers and arrows $n\to m$ are cospans $n \tr{} G\tl{} m$ in $\Hyp{\Sigma}$ (where $n$, $m$ are considered as hypergraphs).
 $\FTerm{\Sigma}$, therefore, is a full subcategory of $\Cospan{\Hyp{\Sigma}}$.
\end{definition}
Explicitly, composition in $\FTerm{\Sigma}$ is defined by pushout as in $\Cospan{\Hyp{\Sigma}}$ and the monoidal product $\tns$ by the coproduct in $\Hyp{\Sigma}$.
 The idea behind the discreteness restriction is that the arrows of the cospan tell what are the ``left and right dangling wires'' in the string diagram encoded by $G$. In pictures, we shall represent $n$ and $m$ as actual discrete graphs --with $n$ and $m$ nodes respectively-- and use number labels (and sometimes colours, whenever available to the reader) to help visualise how they get mapped to nodes of $G$.

Given a signature $\Sigma$, we define a PROP morphism $\synTosem{\cdot} \colon \mathbf{S}_\Sigma \to \FTerm{\Sigma}$.
Since $\syntax{\Sigma}$ is the PROP freely generated by an SMT with no equations, it suffices to define $\synTosem{\cdot}$ on the generators: for each $o \colon n \to m$ in $\Sigma$, we let $\synTosem{o}$ be the following cospan of type $n \to m$

\begin{equation*}\label{eq:phi}
  \scalebox{0.75}{\input{./figures/generator-cospan.tikz}}
\end{equation*}

\begin{example} \label{eq:mimramFormalIntepretation}
The two sides of the PROP rewriting rule~\eqref{eq:mimramYB-prop} (Example~\ref{ex:mimram}) get interpreted as the following cospans in $\FTerm{\Sigma}$
\begin{eqnarray*}
\cgr[height=40pt]{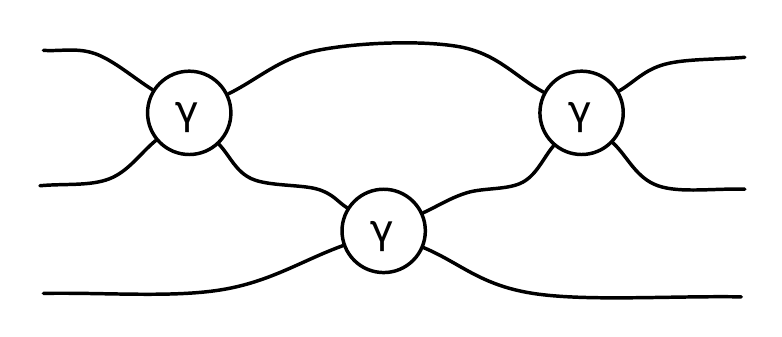} & \quad \xmapsto{\synTosem{\cdot}} \quad &
  \scalebox{0.75}{\input{./figures/yang-lhs-cospan.tikz}}
\\
\cgr[height=40pt]{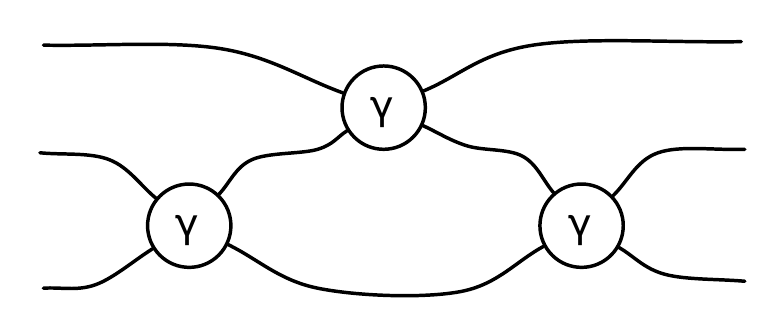} & \quad \xmapsto{\synTosem{\cdot}} \quad &
  \scalebox{0.75}{\input{./figures/yang-rhs-cospan.tikz}}
\end{eqnarray*}
\end{example}

\begin{proposition}[\cite{BGKSZ-partone}]\label{prop:faithful}
$\synTosem{\cdot} \colon \syntax{\Sigma} \to \FTerm{\Sigma}$ is faithful.
\end{proposition}


The encoding $\synTosem{\cdot}$ is an important part of Theorem~\ref{thm:coproduct} below. This is a pivotal result in our exposition, as it serves as a bridge between algebraic and combinatorial structures. Indeed, it provides a presentation, by means of generators and equations, for the PROP $\FTerm{\Sigma}$: the disjoint union of the SMTs of $\syntax{\Sigma}$ and $\frob$.

\begin{theorem}[\cite{BGKSZ-partone}]\label{thm:coproduct}
There is an isomorphism of PROPs $\allTosem{\cdot} \to \syntax{\Sigma} + \frob \tr{ } \FTerm{\Sigma}$.
\end{theorem}

The isomorphism $\allTosem{\cdot}$ is given as the pairing $[\synTosem{\cdot}, \frobTosem{\cdot}] \colon \syntax{\Sigma} + \frob \to \FTerm{\Sigma}$, where $\frobTosem{\cdot} \colon \frob \to \FTerm{\Sigma}$ is the PROP morphism mapping the generators of $\frob$ as follows
\[
\begin{array}{ccccccc}
\Bmult & \mapsto\quad & \input{./figures/csp-mult.tikz}
&\qquad\qquad &
\Bcomult & \mapsto\quad & \input{./figures/csp-comult.tikz} \\[1cm]
\Bunit & \mapsto\quad & \input{./figures/csp-unit.tikz}
&\qquad\qquad&
\Bcounit & \mapsto\quad & \input{./figures/csp-counit.tikz}
\end{array}
\]
Here, $\frob$ is used to model those features of the graph domain that are not part of the syntactic domain,
e.g. the ability of building a ``feedback loop'' around some $\alpha \colon 1 \to 1$ in $\Sigma$
\[ \scalebox{0.75}{\input{./figures/feedback-loop.tikz}} \quad \xmapsto{\allTosem{\cdot}} \quad \scalebox{0.75}{\input{./figures/feedback-loop-csp.tikz}} \]

\subsection{Confluence for rewriting in $\syntax{\Sigma}+\frob$}\label{sec:discrete}

We can use Theorem~\ref{thm:coproduct} to apply results for graphs with interfaces to $\syntax{\Sigma}+\frob$.  To this aim, first we need to interpret string diagrams as graphs with a \emph{single} interface, instead of two as in their usual cospan interpretation. This can be easily achieved by applying the transformation $\rewiring{d}$ (introduced in~\cite{BGKSZ-partone}), which ``rewires'' a syntactic term $d$ of $\syntax{\Sigma}+\frob$ by turning all of the inputs into outputs
\[
  \input{./figures/pf-syntax-d.tikz}
  \qquad \xmapsto{\rewiring{\cdot}} \qquad
  \input{./figures/pf-d-cup.tikz}
\]
Syntactic rewriting with ``rewired'' graphs is equivalent to rewriting with the original ones, in the sense that $d \Rew{\rrule{l}{r}} e$ if and only if $\rewiring{d} \Rew{\rrule{\rewiring{l}}{\rewiring{r}}} \rewiring{e}$.
However, since the rewired rules have only one boundary, they are readily interpreted as hypergraphs with interfaces: if $\allTosem{d} = i \rightarrow G \leftarrow j$, then $\allTosem{\rewiring{d}} = 0 \rightarrow G \leftarrow i + j$, which we may simply write as the graph with interface $G \leftarrow i + j$.

\begin{example} \label{ex:folding} The PROP rewriting system of Example \ref{ex:mimram} consists of just a single rule, 
let us call it $\langle d,  e \rangle$. The resulting DPO rewriting system with interfaces  is then presented 
in Example~\ref{ex:mimramdpo}. Also, Example~\ref{eq:mimramFormalIntepretation} is an intermediate step of this transformation, 
as it shows the cospans $\synTosem{c} = \allTosem{c}$ and $\synTosem{d} = \allTosem{d}$. One can obtain both 
graphs with interfaces $\allTosem{\rewiring{c}}$ and 
$\allTosem{\rewiring{d}}$ by ``folding'' the domain/codomain into the interface of Example~\ref{ex:mimramdpo}.\end{example}

Observe that a rule in the rewrite system $\allTosem{\rewiring{\mathcal{R}}}$ (defined  as $\allTosem{\rewiring{\mathcal{R}}} = \{ \rrule{\allTosem{\rewiring{l}}}{\allTosem{\rewiring{r}}} \mid \rrule{l}{r}  \in \mathcal{R} \}$) just consists of a pair of hypergraphs with a common interface, i.e., it is a DPO rule of the form $L \leftarrow n+m \to R$.
Thus, PROP rewriting in $\FTerm{\Sigma}$ coincides with DPOI rewriting: together with Theorem~\ref{thm:coproduct}, this 
correspondence yields the following result.
%

\begin{theorem}[\cite{BGKSZ-partone}]\label{thm:frobeniusrewriting}
Let $\mathcal{R}$ be a rewriting system on $\syntax{\Sigma}+\frob$. Then
$$d \Rew{\RS} e \mbox{   iff   } \allTosem{\rewiring{d}} \DPOstep{\scriptscriptstyle{\allTosem{\rewiring{\mathcal{R}}}}} \allTosem{\rewiring{e}}.$$

\end{theorem}

One can read Theorem~\ref{thm:frobeniusrewriting} as: DPO rewriting with interfaces is sound and complete for any symmetric monoidal theory with a chosen special Frobenius structure, i.e. one of shape $(\Sigma + \Sigma_F, \mathcal{E} + \mathcal{E}_F)$, with $(\Sigma_F, \mathcal{E}_F)$ the SMT of $\frob$. 
There are various relevant such theories in the literature, such as the ZX-calculus~\cite{Coecke2008}, the calculus of signal flow graphs~\cite{Bonchi2014b}, the calculus of stateless connectors~\cite{Bruni2006} and monoidal computer~\cite{Pavlovic13}.

\medskip

The combination of the result above with Theorem~\ref{th:locconfl} is however not sufficient for
ensuring the decidability of the confluence for a terminating rewriting system $\mathcal{R}$ on $\syntax{\Sigma}+\frob$. Indeed, Theorem~\ref{th:locconfl} and Theorem \ref{thm:frobeniusrewriting} ensure that if all the pre-critical pairs in $\allTosem{\rewiring{\mathcal{R}}}$ are joinable, then the rewriting in $\mathcal{R}$ is confluent. 
However, for the decidability of confluence in $\mathcal{R}$ the reverse is also needed: if one pre-critical pair in $\allTosem{\rewiring{\mathcal{R}}}$ is not joinable, then $\mathcal{R}$ should not be confluent. To conclude this fact, it is enough to check that all pre-critical pairs of $\allTosem{\rewiring{\mathcal{R}}}$ lay in the image of $\allTosem{\rewiring{\cdot}}$, i.e., that they all have discrete interfaces.
The key observation is given  by the lemma below.

\begin{lemma}[Pre-critical pair with discrete interface]\label{lemma:precritical-discrete} Consider a pre-critical pair in $\Hyp{\Sigma}$ as in~\eqref{diag:criticalpairInterface}, Definition~\ref{pre-critical}. 
If both $K_1$ and $K_2$ are discrete, so is the interface $J$.
\end{lemma}
\begin{proof}
For $i=1,2$, since $K_i$ is discrete, the hyperedges of $C_i$ are exactly those of $G_i$ that are not in $f_i(L_i)$.
Since $[f_1,f_2]\colon L_1+L_2 \to S$ is epi, all the hyperedges of $G$ are either in $f_1(L_1)$ or $f_2(L_2)$. Therefore, $J$ cannot contain any hyperedge.
\end{proof}

Since in every rule $L \tl{} K \tr{} R$ in $\allTosem{\rewiring{\mathcal{R}}}$, $K$ is discrete, from Lemma~\ref{lemma:precritical-discrete} and Theorem \ref{th:locconfl} we derive the following result.

\begin{corollary}\label{cor:confluenceFROB}
Confluence is decidable for terminating rewriting systems on ${\syntax{\Sigma}+\frob}$.
\end{corollary}

\begin{proof}
To decide confluence of a rewriting system $\mathcal{R}$ on ${\syntax{\Sigma}+\frob}$, it is enough to check whether all pre-critical pairs in 
$\allTosem{\rewiring{\mathcal{R}}}$ are joinable. Indeed, if all pre-critical pairs are joinable, then $\rewr_{\mathcal{R}}$ is confluent by 
Theorems \ref{th:locconfl} and \ref{thm:frobeniusrewriting}. For the other direction, suppose that there exists a pre-critical pair 
${}_{\scriptscriptstyle{\allTosem{\rewiring{\mathcal{R}}}}}\LDPOstep{} (S\tl{}J ) \DPOstep{\scriptscriptstyle{\allTosem{\rewiring{\mathcal{R}}}}}$ that is not joinable. By construction, in every rule $L \tl{} K \tr{} R$ in 
$\allTosem{\rewiring{\mathcal{R}}}$, $K$ is discrete. Therefore, by Lemma \ref{lemma:precritical-discrete}, also $J$ is discrete. This is the key fact to entail that there exists $d$ in $\syntax{\Sigma}+\frob$, such that $\allTosem{\rewiring{d}} = (S\tl{}J)$.  By Theorem \ref{thm:frobeniusrewriting}, $d$ witnesses that  $\rewr_{\mathcal{R}}$ is not confluent.

Now, if $\mathcal{R}$ is terminating, then by Theorem~\ref{thm:frobeniusrewriting} also $\allTosem{\rewiring{\mathcal{R}}}$ is terminating. The latter is also computable and therefore joinability of pre-critical pairs of $\allTosem{\rewiring{\mathcal{R}}}$ can easily be decided by following the steps in the second part of the proof of Corollary \ref{cor:comput}. 
\end{proof}

\section{Confluence for PROP rewriting without Frobenius structure}\label{sec:confluenceSMC}


The presence of a chosen Frobenius structure simplifies the connection between syntactic PROP rewriting and hypergraph rewriting in two ways. The first is that Frobenius algebras give a natural, syntactic analogue to hypergraph vertices that are connected to many different hyperedges at once. These correspond to wires `splitting' and `merging', which are exactly captured by the Frobenius algebra.
The second, and perhaps more notable freedom provided by Frobenius algebras is the ability to interpret feedback loops, and in particular hypergraphs and hypergraph rewriting that ignores any kind of acyclicity constraint. We have shown in Part II, and the reason why is shortly recalled here in Section~\ref{sec:convex}, that the absence of Frobenius algebras requires us not only to restrict the types of hypergraphs we consider, but also which matches correspond to syntactically-sound rewriting steps. This restriction on allowed matches has significant consequences in proving local confluence by critical pair analysis, which we will address in this section.


\subsection{Monogamous acyclic hypergraphs}\label{sec:ma-hypergraphs}

We first recall from~\cite{BGKSZ-parttwo} a combinatorial characterisation of the image of $\synTosem{\cdot}$. It is based on a few preliminary definitions. We call a sequence of hyperedges $e_1, e_2, \ldots, e_n$ a \textit{(directed) path} if at least one target of $e_k$ is a source for $e_{k+1}$ and a \textit{(directed) cycle} if additionally at least one target of $e_n$ is a source for $e_1$.
The \emph{in-degree} of a node $v$ in an hypergraph $G$ is the number of pairs $(h,i)$ where $h$ is an hyperedge with $v$ as its $i$-th target. Similarly, the \emph{out-degree} of $v$ is the number of pairs $(h,j)$ where $h$ is an hyperedge with $v$ as its $j$-th source. We call \emph{input} nodes those with in-degree $0$, \emph{output} nodes those with out-degree $0$, and \emph{internal} nodes the others. We write $\inp{G}$ for the set of inputs and $\out{G}$ for the set of outputs.

\begin{definition}
  A hypergraph $G$ is \textit{monogamous acyclic} (ma-hypergraph) if
    \begin{enumerate}
    \item it contains no cycle (acyclicity) and
    \item every node has at most in- and out-degree $1$ (monogamy).
  \end{enumerate}
  A cospan $n \tr{f} G \tl{g} m$ in $\FTerm{\Sigma}$ is monogamous acyclic (ma-cospan) when $G$ is an ma-hypergraph, $f$ is mono and its image is $\inp{G}$, and $g$ is mono and its image is $\out{G}$.
\end{definition}






\begin{theorem}[\cite{BGKSZ-parttwo}]\label{thm:charactImage} $n \tr{} G \tl{} m$ in $\FTerm{\Sigma}$ is in the image of $\synTosem{\cdot}$ iff it is
an ma-cospan. \end{theorem}

We call a hypergraph with interface $G \tl{} J$ monogamous acyclic (ma-hypergraph with interface) if it is of the form $G \tl{[i,o]} n + m$ for an ma-cospan $n \tr{i} G \tl{o} m$, up to isomorphism between $J$ and $n+m$. Note that such an ma-cospan (and hence the ma-hypergraph with interface) may be uniquely fixed, so we can use the two representations interchangeably.


Finally, we say that a rule $L \tl{} i+j  \tr{} R$ is an ma-rule if $i \tr{} L \tl{} j$ and $i \tr{} R \tl{} j$ are ma-cospans.

\subsection{Convex rewriting and soundness}\label{sec:convex}
%
%
We are now in position to interpret PROP rewriting for $\syntax{\Sigma}$ in DPO rewriting
for ma-hypergraphs with interfaces, via the mapping
that takes string diagrams to ma-hypergraphs with interfaces.
Unfortunately, as shown in Part II~\cite{BGKSZ-parttwo}, this interpretation is generally \emph{unsound}.
There are several things that can go wrong in the absence of Frobenius structure, as illustrated in the next
two examples from Part II. These motivate our restrictions to PROP rewriting systems that make the
interpretation sound, as presented in the sections below.


First, as noted in Part I~\cite{BGKSZ-partone}, a DPOI rule can have multiple pushout complements when it is not left-linear, only some of which make sense without Frobenius structure.

\begin{example}\label{ex:unsoundcontext}
Consider $\Sigma = \{ \alpha_1 \colon 0 \to 1, \alpha_2 \colon 1\to 0, \alpha_3 \colon 1 \to 1\}$ and the PROP rewriting system $\mathcal{R} = \left\{\ \begin{tikzpicture}
	\begin{pgfonlayer}{nodelayer}
		\node [style=none] (1) at (-1, 0) {};
		\node [style=none] (2) at (1, 0) {};
	\end{pgfonlayer}
	\begin{pgfonlayer}{edgelayer}
		\draw (1.center) to (2.center);
	\end{pgfonlayer}
\end{tikzpicture}
\  \Rew{} \ \begin{tikzpicture}
	\begin{pgfonlayer}{nodelayer}
		\node [style=small box] (0) at (0, 0) {$\alpha_3$};
		\node [style=none] (1) at (-1.25, 0) {};
		\node [style=none] (2) at (1.25, 0) {};
	\end{pgfonlayer}
	\begin{pgfonlayer}{edgelayer}
		\draw (1.center) to (2.center);
	\end{pgfonlayer}
\end{tikzpicture}
\ \right\}$ on $\syntax{\Sigma}$. Its interpretation in $\FTerm{\Sigma}$ is given by the rule
\[ \scalebox{0.75}{\input{./figures/unsound-context-rule.tikz}} \]
The rule is not left-linear and therefore pushout complements are not necessarily unique for the application of this rule. For example, the following pushout complement yields a rewritten graph that can be interpreted as an arrow in an SMC
\[ \scalebox{0.75}{\input{./figures/sound-context-dpo.tikz}} \]
On the other hand, if we choose a different pushout complement, we obtain a rewritten graph that does not look like an SMC morphism
\[ \scalebox{0.75}{\input{./figures/unsound-context-dpo.tikz}} \]

The different outcome is due to the fact that $f$ maps $0$ to the leftmost and $1$ to the rightmost node, whereas $g$ swaps the assignments. Even though both rewriting steps could be mimicked at the syntactic level in $\syntax{\Sigma} + \frob$, the second hypergraph rewrite yields a hypergraph that is illegal for $\mathcal{R}$ in $\syntax{\Sigma}$. In particular, the rewritten graph in the second derivation is not monogamous: the outputs of $\alpha_1$ and $\alpha_3$ and the inputs of $\alpha_2$ and $\alpha_3$ have been glued together by the right pushout.
\end{example}

Next we can see that, even if a DPOI rewriting step yields a string diagram that can be expressed without Frobenius structure, it could be the case that equation itself cannot be proven in the SMT without introducing a feedback loop.

\begin{example}\label{ex:unsound}
Consider a $\Sigma = \{ e_1 \: 1 \to 2, {e_2 \: 2 \to 1}, {e_3 \: 1 \to 1} , e_4 \: 1 \to 1\}$ and the following rewriting rule in $\syntax{\Sigma}$
  \begin{eqnarray}
    \rrule{ \ \input{./figures/unsound-lhs.tikz}\ \colon 2 \to 2 \ }{\ \ \input{./figures/unsound-rhs.tikz}\ \colon 2\to 2 } \label{eq:counterexSoundness1}
  \end{eqnarray}
  Left and right side are interpreted in $\FTerm{\Sigma}$ as cospans
  \[
    \input{./figures/unsound-lhs-cospan.tikz}
    \qquad
    \input{./figures/unsound-rhs-cospan.tikz}
  \]
  We introduce another diagram $c \: 1\to 1$ in $\syntax{\Sigma}$ and its interpretation in $\FTerm{\Sigma}$
  \[
    \input{./figures/unsound-target.tikz}
    \quad \xmapsto{\synTosem{\cdot}} \quad
    \input{./figures/unsound-target-cospan.tikz}
  \]
  Now, rule \eqref{eq:counterexSoundness1} cannot be applied to $c$, even modulo the SMC equations. However, their interpretation yields a DPO rewriting step in $\FTerm{\Sigma}$ as below
  \[
    \scalebox{0.75}{\input{./figures/unsound-dpo.tikz}}
  \]
  Observe that the leftmost pushout above \emph{is} a boundary complement: the input-output partition is correct. Still, the rewriting step cannot be mimicked at the syntactic level using rewriting modulo the SMC laws. That is because, in order to apply our rule, we need to deform the diagram such that $e_3$ occurs outside of the left-hand side. This requires moving $e_3$ either before or after the occurence of the left-hand side in the larger expression, but both of these possibilities require a feedback loop
  \[ \input{./figures/e3-before.tikz} \qquad\qquad \input{./figures/e3-after.tikz} \]
  Hence, if the category does not have at least a traced symmetric monoidal structure~\cite{Joyal_tracedcategories}, there is no way to apply the rule.
\end{example}


%

The two examples motivate the definition of \emph{convex} rewriting, as a restriction of DPOI rewriting that rules out the above counterexamples and ensures soundness. We briefly recall the relevant definitions from Part II~\cite{BGKSZ-parttwo}, referring to the discussion therein for more examples and properties of convex rewriting. As a preliminary step, we need to introduce a suitable restriction of the notion of pushout complement, called \emph{boundary complement}.

\begin{definition}[Boundary complement] \label{def:boundarycomplement}
    For ma-cospans $i \xrightarrow{a_1} L \xleftarrow{a_2} j$ and $n \xrightarrow{b_1} G \xleftarrow{b_2} m$
    and mono $f : L \to G$, a pushout complement as depicted in $(\dagger)$ below
    \[ \xymatrix@C=50pt@R=20pt{
        L \ar[d]_f \ar@{}[dr]|{(\dagger)} & {i + j} \ar[l]_{a = [a_1,a_2]}
              \ar[d]^{c = [c_1,c_2]}
        \ar@{}[dl]|(.8){\text{\large $\urcorner$}}
         \\
       G & L^\perp \ar[l]^{g} \\
       & n+m \ar[ul]^{[b_1,b_2]} \ar@{-->}[u]_{[d_1,d_2]}
       } \]
    is called a \textit{boundary complement} if $[c_1,c_2]$ is mono and there exist
    $d_1\: n\to L^\perp$ and $d_2\: m \to L^\perp$ making the above triangle commute
    and such that
    \begin{equation}\label{eq:boundary} j+n \xrightarrow{[c_2,d_1]} L^\perp \xleftarrow{[c_1,d_2]} m+i \end{equation}
    is an ma-cospan.
\end{definition}

Note that boundary complements are unique when they exist~\cite{BGKSZ-parttwo}, and that the pushout complement of Example~\ref{ex:unsoundcontext} is not a boundary complement.

The next definition is a restriction on the possible matches, and rules out the other counterexample (Example~\ref{ex:unsound}).

\begin{definition}[Convex match] \label{def:convexampleatching}
We call $m : L \to G$ in $\Hyp{\Sigma}$ a \textit{convex match} if it is mono and its image $m[L]$ is convex, that is, for any nodes $v, v'$ in $m[L]$ and any  path $p$ from $v$ to $v'$ in $G$,
every hyperedge in $p$ is also in $m[L]$.
\end{definition}

We now have all the ingredients to recall the notion of convex rewriting step, which is essentially a DPOI rewriting step relying on a boundary complement and a convex match.

\begin{definition}\label{def:convex-dpoi}
Given $n \rightarrow D \leftarrow m$ and $n \rightarrow E \leftarrow m$ ma-cospans, \emph{$D$ rewrites convexely into $E$ with interface $n+m$} --notation
$(D \tl{} n+m) \rigidDPOstep{\mathcal{R}} (E \tl{} n+m)$-- if there exist ma-rule $L \tl{} i+j  \tr{} R$ in $\mathcal{R}$,
object $C$, and morphisms such that the diagram below commutes and the squares are pushouts
\begin{equation}\label{eq:dpo3}
\raise25pt\hbox{$
\xymatrix@R=15pt@C=20pt{
L \ar[d]_{f}   &  i+j \ar[d]
 \ar@{}[dl]|(.8){\text{\large $\urcorner$}}
 \ar@{}[dr]|(.8){\text{\large $\ulcorner$}}
 \ar[l]_{[a_1,a_2]} \ar[r]^{[b_1,b_2]}  & R \ar[d] \\
 D &  C \ar[l] \ar[r]  & E \\
&  n+m \ar[u] \ar[ur]_{[p_1,p_2]}  \ar[ul]^{[q_1,q_2]}
}$}
\end{equation}
and the following conditions hold
 \begin{itemize}
 \item $f \: L \to D$ is a convex match, and
 \item $i+j \to C \to D$ is a boundary complement in the leftmost pushout.
\end{itemize}
\end{definition}

Note that in the definition above we implicitly assume that $i \rightarrow L \leftarrow j $ is an ma-cospan.

\subsection{Failure of naive critical pair analysis for convex rewriting}\label{sec:failure-convex}

The main issue with using the critical pair analysis technique delineated before is that convexity is not preserved by clipping. That is, we can have critical overlaps $H$ which form non-convex subgraphs $H \subseteq G$. Hence, it could be the case that a branching $H_1 \LDPOstep{} H \DPOstep{} H_2$ is joinable using convex rewriting starting from $H$, but the lifted branching $G_1 \LDPOstep{} G \DPOstep{} G_2$ will not be joinable.

This is what happens in the following example, taken from Part II~\cite{BGKSZ-parttwo}.

\begin{example} \textit{Frobenius semi-algebras} are Frobenius algebras lacking the unit and counit equations. That is, they are the free PROP generated by the signature
\[
  \left\{\ \mu := \Wmult, \delta := \Wcomult\ \right\}
\]
modulo the following equations
\begin{equation}\label{eq:fsa-rules}
  \input{./figures/fsa-rules.tikz}
\end{equation}
which can be represented as the following rewriting system
\begin{center}
  \scalebox{0.7}{\input{./figures/fsa-rules-dpo.tikz}}
\end{center}

Suppose we start to consider the pre-critical pairs for this system. The following one is clearly joinable, since it involves parallel (and in fact totally disjoint) applications of \FS{3} and \FS{4}
\[ \scalebox{0.75}{\input{./figures/convex-clipping-fail.tikz}} \]

However, if we consider the middle graph in a larger context, e.g.
\[ \scalebox{0.75}{\input{./figures/convex-clipping-fail2-alt.tikz}} \]
it no longer becomes joinable by convex rewriting. For example, suppose we apply \FS{3} on the larger graph above
\[ \scalebox{0.75}{\input{./figures/convex-clipping-fail3.tikz}} \]
Then we are stuck: \FS{4} no longer has a convex matching, because applying \FS{3} introduced a new path from the output 5 to the input 0. Convexity guarantees we will not introduce cycles, and indeed in this case applying \FS{4} would introduce a new path from 0 to 5, and hence a cycle.



\end{example}


We present two solutions to this problem. The first is to put a strong restriction, called \textit{left-connectedness}, on the rewriting systems being considered. The second is to develop a more in-depth notion of critical pair analysis, which accounts for the context-sensitivity of convex rewriting using \textit{formal path extensions}.

They both rely on a more refined notion of pre-critical pair. One cannot simply reuse Definition~\ref{pre-critical}, as we want to enforce that the common source $S \tl{} J$ (cf. \eqref{diag:criticalpairInterface}) of the two derivations is an ma-hypergraph with interfaces, so that it is in the image of $\allTosem{\rewiring{\cdot}}$ and we can reason about pre-critical pairs `syntactically' in $\syntax{\Sigma}$.
However, while Lemma~\ref{lemma:precritical-discrete} guarantees that this is always the case for rewriting systems on $\syntax{\Sigma}+\frob$,
with Definition~\ref{pre-critical}, this is not guaranteed for $\syntax{\Sigma}$, as shown by the example below.

\begin{example} We concoct a pre-critical pair by instantiating~\eqref{diag:criticalpairInterface} as
shown below
\[ 
\cgr[height=1.5cm]{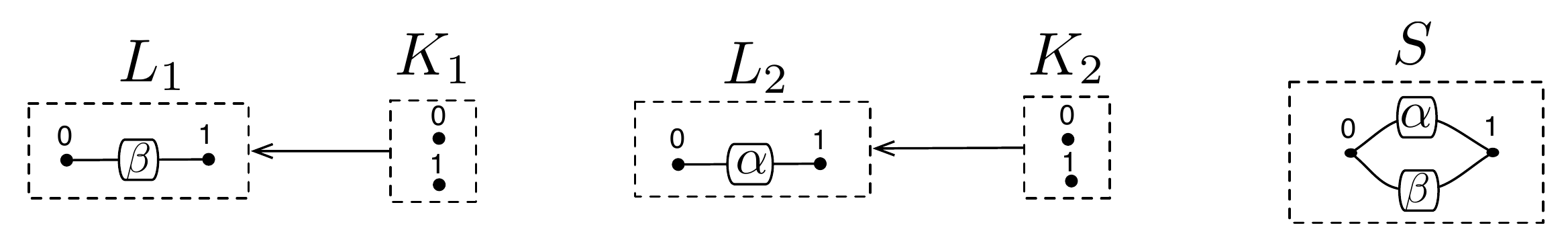}
\]
Although $L_1 \tl{} K_1$ and $L_2 \tl{} K_2$ are left-hand sides of left-connected rules, $S$ is \emph{not} monogamous, thus this pre-critical pair does not correspond to anything syntactic in $\syntax{\Sigma}$.
\end{example}

Recall from the end of Section~\ref{sec:ma-hypergraphs} that an ma-hypergraph with interface is required to have an interface corresponding exactly to the inputs and outputs of the ma-hypergraph. Because of this, using the notion of pre-critical pair from Definition~\ref{pre-critical} may yield too many nodes in the interface.

\begin{example}\label{ex:badpullback} Here is an example, where two rules match in an ma-hypergraph $G$, but the interface contains one extra node $4$ which is neither an input nor an output of $G$
\[\cgr{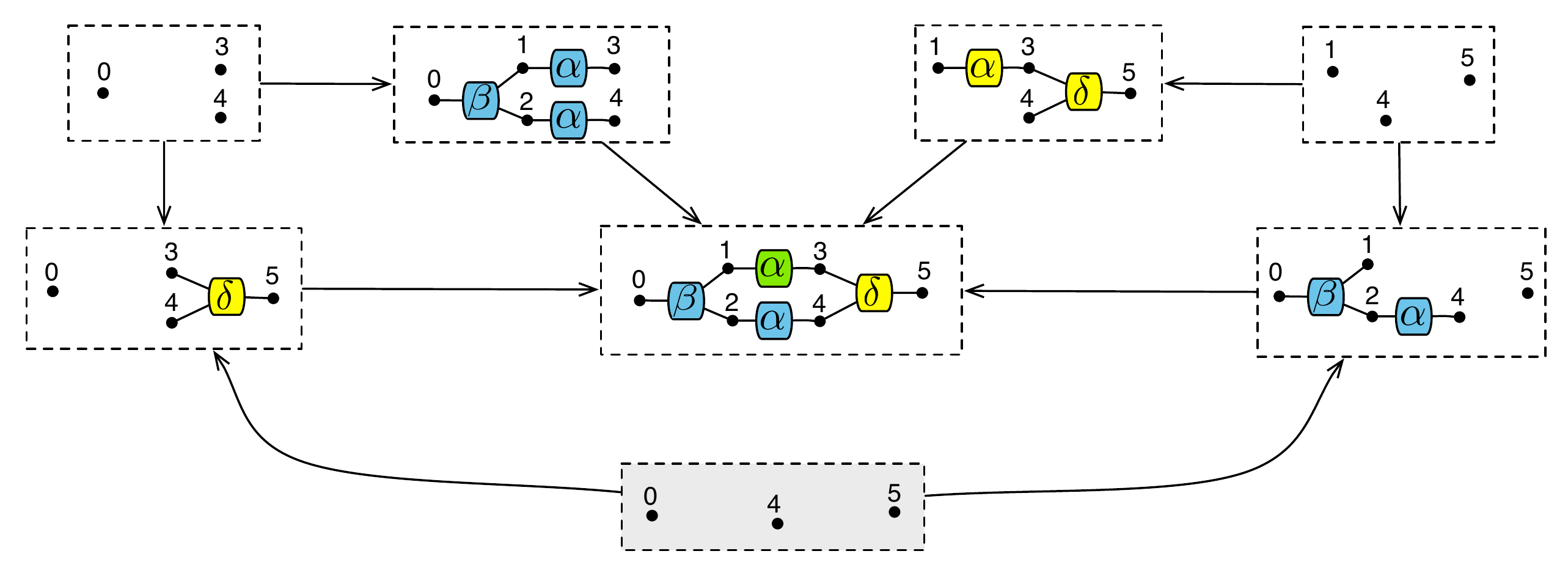}\]
\end{example}

Motivated by these two examples, we give the following definition.

\begin{definition}[Ma-pre-critical pair] Let $\mathcal{R}$ be a rewrite system consisting of ma-rules $L_1 \tl{} K_1 \tr{} R_1$ and $L_2 \tl{} K_2 \tr{} R_2$.
Consider two derivations with source $S\tl{} J$ 
\begin{equation}\label{diag:macriticalpairInterface}
\vcenter{
    \xymatrix@R=12pt{
    R_1 \ar[d] & \ar[l] K_1 \dlcorner  \ar[d] \ar[r] & L_1 \ar@{}[dr]|(.8){\text{\large $\ulcorner$}\qquad\quad} \ar[dr]^{f_1} && \ar[dl]_{f_2} \ar@{}[dl]|(.8){\qquad\quad\text{\large $\urcorner$}} L_2 & \ar[l] K_2 \ar[d] \drcorner\ar[r] & R_2 \ar[d] \\
    H_1  & \ar[l] C_1 \ar[rr] & & S  & & \ar[ll] C_2  \ar[r] & H_2 \\
    &&& \ar@/^/[ull] J \ar@{}[u]|{(\dagger)} \ar@/_/[urr] & & &
    }
}
\end{equation}
We say that $(H_1\tl{} J) \LDPOstep{} (S\tl{}J) \DPOstep{} (H_2\tl{} J)$ is an \emph{ma-pre-critical pair} if
 $[f_1,f_2] \colon L_1 + L_2 \to S$ is epi, $(\dagger)$ is a commuting diagram, and
 $S\tl{} J$ is an ma-hypergraph with interface;
it is \emph{joinable} if there exists an ma-hypergraph with interface $W\tl{}J$ such that
$(H_1\tl{}J) \DPOstep{}^* (W\tl{}J) {\,\,{}^*\!\!\LDPOstep{}\,\,\,} (H_2\tl{}J)$.
\end{definition}

Comparing it with Definition~\ref{pre-critical}, we are dropping the requirement that $(\dagger)$ is a pulback. However, note that 
up to an isomorphic choice of $J$, there is at most one ma-hypergraph with interface $S\tl{} J$ . Indeed, all 
ma-pre-critical pairs are also pre-critical pairs, and if $I$ is the pullback along $C_1 \rightarrow S \leftarrow C_2$, the 
uniquely induced monomorphism  $I\tl{} J$ just weeds out from $I$ those items whose image is
neither an input nor an output of $S$. 

\subsection{Confluence for left-connected rewriting in $\syntax{\Sigma}$}\label{sec:SymMonDpoRewriting}

\begin{definition}\label{def:left-connected} An ma-hypergraph $G$ is \emph{strongly connected} if for every input $x \in \inp{G}$ and output $y \in \out{G}$ there exists a path from $x$ to $y$. A DPO system with interfaces is called \textit{left-connected} if it is left-linear, every rule is an ma-rule and its left-hand side is strongly connected. We call a PROP rewriting system $\mathcal{R}$ on $\syntax{\Sigma}$ left-connected if $\allTosem{\rewiring{\mathcal{R}}}$ is left-connected.
\end{definition}
Non-commutative bimonoids (Example~\ref{exm:props}\ref{ex:ncbialgebras}, see also \S\ref{sec:confluenceBimonoids} below) and the Yang-Baxter rule of Example \ref{ex:mimram} are examples of left-connected rewriting systems.

Intuitively, in Definition~\ref{def:left-connected} strong connectedness prevents matches leaving ``holes'', as in Example~\ref{ex:unsound}, whereas left-linearity guarantees uniqueness of the pushout complements, and prevents the problem in Example~\ref{ex:unsoundcontext}. We are then able to prove the following.

%
\begin{theorem}[\cite{BGKSZ-parttwo}]\label{thm:stronglyconnected}
Let $\mathcal{R}$ be a left-connected rewriting system on $\syntax{\Sigma}$. Then,
\begin{enumerate}
\item if $d \Rew{\RS} e$ then $\allTosem{\rewiring{d}}  \rigidDPOstep{\allTosem{\rewiring{\mathcal{R}}}} \allTosem{\rewiring{e}}$;
\item if $\allTosem{\rewiring{d}} \rigidDPOstep{\allTosem{\rewiring{\mathcal{R}}}} \allTosem{\rewiring{e}}$ then $d \Rew{\RS} e$.
\end{enumerate}
\end{theorem}

\begin{remark}
%
Note 
that for such rewriting systems the further restriction of
left-linearity is not particularly harmful, confluence-wise.
Indeed, an ma-hypergraph with interface $G \leftarrow J$ is not mono iff $G$ has one node that is both input and output, i.e., an isolated node. A rule with a strongly connected $L \leftarrow K$ is not left-linear precisely when $L$ is discrete, with a single node. Such a rule cannot be part of a terminating system, that is, one where local confluence implies confluence.\end{remark}
%
%

The above theorem allows us to use DPOI rewriting as a mechanism for rewriting $\syntax{\Sigma}$.

%
%
%
%
%
We could now recast in this setting the considerations on parallel and critical pairs, as well as on joinability,
as given in Definition~\ref{parallelpairs} and Proposition~\ref{confparallelpairs}, respectively.
We move instead directly to state the confluence theorem for left-connected systems.

 \begin{theorem}[Local confluence for left-connected systems]
 \label{th:locconflwellformed}
For a left-connected DPO system with interfaces,
if all ma-pre-critical pairs are joinable then rewriting is locally confluent: given an ma-hypergraph with interface $G_0 \tl{} I$
and $(G_{1,1} \tl{} I) \LDPOstep{} (G_0 \tl{} I) \DPOstep{} (G_{1,2} \tl{} I)$,
 there exists an ma-hypergraph with interface $W \tl{} I$ such that
 \[
\xymatrix@R=5pt@C=10pt{
&  \ar@{~>}[dl]  (G_0 \tl{} I) \ar@{~>}[dr] & \\
(G_{1,1} \tl{} I) \ar@{~>}[dr]^>>>{*} && (G_{1,2} \tl{} I) \ar@{~>}[dl]_>>>{*} \\
&(W \tl{} I).&
}
\]
 \end{theorem}

The proof of Theorem~\ref{th:locconflwellformed} follows steps analogous to the one of Theorem~\ref{th:locconfl}. The essential difference is that ma-pre-critical pairs now have interfaces that are not necessarily pullbacks. The assumption of left-connectedness is nevertheless enough to ensure that the fundamental pieces, namely Constructions~\ref{constr:extraction} and~\ref{constr:embedding}, can be reproduced.

\begin{corollary}\label{cor:confluenceSMTs}
Let $\mathcal{R}$ be a terminating left-connected rewriting system on $\syntax{\Sigma}$. Then confluence of $\rewr_{\mathcal{R}}$ is decidable.
\end{corollary}
\begin{proof} By Theorem~\ref{thm:stronglyconnected} and~\ref{th:locconflwellformed}, it is enough to check whether pre-critical pairs in
$\allTosem{\rewiring{\mathcal{R}}}$ are joinable. This is decidable since $\mathcal{R}$ is terminating and $\allTosem{\rewiring{\mathcal{R}}}$ is computable. \end{proof}

\begin{example}  \label{eq:mimramConfluent} The PROP rewriting system $\mathcal{R}$ of Example~\ref{ex:mimram} is left-connected. Once interpreted as the DPO rewriting system with interfaces of Example~\ref{ex:mimramdpo}, we can do critical pair analysis. The ma-pre-critical pair below (where the middle grey graph acts as the interface for the rewriting steps) is not joinable, meaning that $\mathcal{R}$ is not confluent

\[\cgr[width=12.5cm]{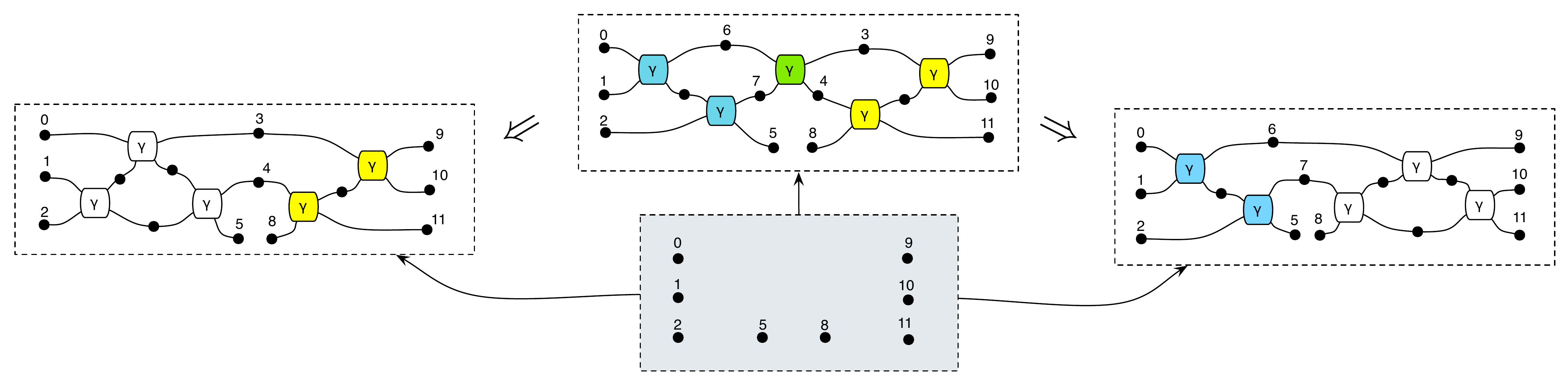}\]
We emphasise that the decision procedure relies on the fact that there are only \emph{finitely many} pre-critical pairs to consider, the above one being the only one to feature a 
non-trivial overlap of rule applications. This is in contrast with a naive, `syntactic' analysis, which as we observe in Example~\ref{ex:mimram} yields infinitely many pre-critical 
pairs for $\mathcal{R}$.
%
\end{example}

\subsection{Convex critical pair analysis via formal path extensions}\label{sec:path-extensions}
\label{convex confluence}

It is natural to ask whether we can extend critical pair analysis for convex rewriting beyond left-connected systems.
It turns out that this is true, but the usual checking of critical pairs does not suffice: they need to be checked in a variety of contexts
to account for the possible existence of paths from an output of the critical pair to an input.
However, while it might seem necessary to check infinitely many contexts to account for every way
a critical pair can be embedded in a larger graph, we get around this problem by considering
\emph{formal path contexts}. These abstract over the particular graph in which a critical pair is embedded,
and only capture whether certain paths exist.




\begin{definition}\rm
For an ma-hypergraph $G$ and a mono $m : G \to H$, the \textit{path relation} of $m$, $R_m \subseteq \out{G} \times \inp{G}$ is defined by letting $(y,x) \in R_m$ if and only if there is a path from the image of $y$ to
the image of $x$ in $H$.
 We say a mono $m' : G \to H'$ \textit{path-covers} $m$, written $m \lesssim m'$, if $R_m \subseteq R_{m'}$.
\end{definition}

Any morphism path-covers itself and path-covering is transitive, so $\lesssim$ is a pre-order (but not a partial order). We will write $m \sim m'$ for $m \lesssim m'$ and $m' \lesssim m$.

\begin{lemma}\label{lem:path-pre}
For an ma-hypergraph $G$ and monos $m : G \to H, m' : G \to H'$ such that $m \lesssim m'$, we have for any mono $k : K \to G$, $m \circ k \lesssim m' \circ k$.
\end{lemma}

\begin{proof}
  If there is a path from an output to an input of $m \circ k [K]$ in $H$, it must split into 3 parts: a path from an output of $m \circ k[K]$ to an output of $m[G]$, a path from an output of $m[G]$ to an input of $m[G]$, and a path from an input of $m[G]$ to an input of $m \circ k[K]$. The first and third parts will also be present in the image of $m' \circ k$, and the second part will be whenever $m \lesssim m'$. Hence $m \circ k \lesssim m' \circ k$.
\end{proof}

We extend $\Sigma$ with 3 new formal path generators $\mathcal P = \{\ \input{./figures/fpath-mul.tikz}, \input{./figures/fpath-comul.tikz}, \begin{tikzpicture}
	\begin{pgfonlayer}{nodelayer}
		\node [style=vsmall box] (0) at (0, 0) {};
		\node [style=none] (5) at (-0.5, 0) {};
		\node [style=none] (6) at (0.5, 0) {};
	\end{pgfonlayer}
	\begin{pgfonlayer}{edgelayer}
		\draw (5.center) to (0);
		\draw (0) to (6.center);
	\end{pgfonlayer}
\end{tikzpicture}
 \}$, and introduce a family of monos which can produce any path relation.

\begin{definition}\label{def:path-ext}\rm
  For an ma-hypergraph $G$ and a binary relation $R \subseteq \out{G} \times \inp{G}$, a mono $p : G \to P$ is called a \textit{path extension} if $P$ consists of $p[G] \cong G$, augmented by additional vertices and $\mathcal P$-labelled hyperedges such that there is a path from an output $y \in j$ to an input $x \in i$ if and only if $(y, x) \in R$.
\end{definition}

It follows by construction that $R = R_p$. Note that there is more than one way to construct a path extension for a given $R$, but if $R = R_p = R_q$ then $p \sim q$.



\begin{lemma}\label{lem:path-ext-rewrite}
Let $L \leftarrow K \rightarrow R$ be a left-linear ma-rule in $\Hyp{\Sigma}$, $p: G' \to P$ a path extension,
and $m : L \to P$ a convex match. Then, $m$ factors as $L \xrightarrow{m'} G' \xrightarrow{p} P$,
and the convex rewrite of $G'$ at $m'$ extends to a convex rewriting step of $P$ at $m$ as follows
  \begin{equation}\label{eq:path-cut}
    \xymatrix{
      L\ar[d]_{m'} & \ar[l] K\ar[d] \ar[r] & R\ar[d]^{n'} \\
      G'\ar[d]_{p} & \ar[l] C'\ar[d] \ar[r] & H'\ar[d]^{q} \\
      P & \ar[l] D \ar[r] & Q \\
    }
  \end{equation}
  and furthermore $q$ is a path extension.
\end{lemma}

\begin{proof}
  Because $L$ contains no $\mathcal P$-hyperedges, the image of every hyperedge in $L$ under $m$ must be in the image of $G'$ under $p$. Furthermore, by left-linearity $L$ contains no isolated vertex, so every vertex in the image of $m$ is in the image of $G'$ under $p$. Hence $m$ factors as $p \circ m'$, as required.

  The top pushouts in \eqref{eq:path-cut} are constructed as a convex DPO rewriting step. The bottom-left pushout is constructed as a pushout complement, which exists because $m$ satisfies the gluing conditions with respect to $K$, so $p$ satisfies the gluing conditions with respect to $C'$ (which contains $K$). The bottom-right square is a pushout. This corresponds to the original rewrite $P \DPOstep{} Q$ by uniqueness of the pushout complement $D$.

  It only remains to show that $q$ is a path extension. This follows from the fact that any $\mathcal P$-hyperedges in the pushout yielding $Q$ must come from $D$.
\end{proof}

\begin{lemma}\label{lem:path-ext-rewrite2}
Let $L \leftarrow K \rightarrow R$, $p$, $q$, $m'$, and $n'$ be given as in Lemma~\ref{lem:path-ext-rewrite}.
Then, for any mono $k : G' \to G$ such that $k \lesssim p$ the convex rewriting of $G'$ at $m'$
extends to a convex rewriting of $G$ at $k \circ m'$ as follows, where $l \lesssim q$
  \begin{equation}\label{eq:path-glue}
    \xymatrix{
      L\ar[d]_{m'} & \ar[l] K\ar[d] \ar[r] & R\ar[d]^{n'} \\
      G'\ar[d]_{k} & \ar[l] C'\ar[d] \ar[r] & H'\ar[d]^{l} \\
      G & \ar[l] C \ar[r] & H \\
    }
  \end{equation}
\end{lemma}

\begin{proof}
  The bottom pushout squares are constructed from the rewrite of $G'$ as in Construction~\ref{diag:embedding}. We first need to show that $k \circ m'$ is convex. If that were not the case, there would be a path from an output of the image of $L$ to an input in $G$. But then, since $k \lesssim p$, we have by Lemma~\ref{lem:path-pre} that $k \circ m' \lesssim p \circ m'$. But then there is a path from an output of the image of $L$ in $P$ to an input, which contradicts convexity of $m = p \circ m'$. Hence $k \circ m'$ is convex.

  It only remains to show that $l \lesssim q$. Inspecting the bottom pushout squares of \eqref{eq:path-cut}, we note that, because of monogamy of $G'$, it must be the case that any path from an output to an input of the image $G'$ in $P$ must be in $C'$. Hence, the same path will be in the image of $H'$ in $Q$, so $R_p \subseteq R_q$. By the symmetry of the DPO construction, it is also the case that $R_q \subseteq R_p$ so $R_p = R_q$.

  Applying the same argument to the bottom pushout squares of \eqref{eq:path-glue}, we conclude that $R_l = R_k$. Since $k \lesssim p$, we have $R_l = R_k \subseteq R_p = R_q$, so $l \lesssim q$.
\end{proof}

\begin{lemma}\label{lem:lift-joins}
For ma-rules $L_1 \leftarrow K_1 \rightarrow R_1$ and $L_2 \leftarrow K_2 \rightarrow R_2$ and an ma-pre-critical pair
$(G'_{1,1} \leftarrow J') \LDPOstep{} (G'_0 \leftarrow J') \DPOstep{} (G'_{1,2} \leftarrow J')$, let $k : G'_0 \to G_0$
be a mono such that the induced matches $L_1 \to G'_0 \to G_0$ and $L_2 \to G'_0 \to G_0$ are convex and
$p : G'_0 \to P_0$ a path extension such that $k \lesssim p$. Then, we can obtain the following rewrites by extending
the critical pair along $k$ and $p$, respectively
  \[
    (G_{1,1} \leftarrow J) \LDPOstep{} (G_0 \leftarrow J) \DPOstep{} (G_{2,1} \leftarrow J)
    \qquad
    (P_{1,1} \leftarrow K) \LDPOstep{} (P_0 \leftarrow K) \DPOstep{} (P_{2,1} \leftarrow K)
  \]
  If the right branching is joinable by convex rewriting, then so is the left one.
\end{lemma}

\begin{proof}
  We apply essentially the same technique as the proof of Theorem~\ref{th:locconfl}, except that we additionally need to show that, when we extend derivations from the critical pair $G'_0$ to the full graph $G_0$ following Construction~\ref{diag:embedding}
  \[     \scalebox{.7}{
     \xymatrix@R=20pt@C=7pt{
       & & & L_1 \ar[ld] & K_1 \ar[d]\ar[l] \ar[r] & R_1\ar[rd] & & L_2 \ar[ld] & K_2 \ar[d]\ar[l] \ar[r] & R_2 \ar[rd] & & \dots & & L_n \ar[ld] & K_n \ar[d] \ar[l] \ar[r] & R_n \ar[rd] \\
     J' 
        \ar@{}[rrd]|(0.5){(\dagger)} \ar[rr] \ar[d]
       & & G_0' \ar[d] & & C_1' \ar@{}[dll]|(0.50){(\gamma_1)} \ar@{}[drr]|(0.50){(\delta_1)}  \ar[d] \ar[ll] \ar[rr] & & G_1' \ar[d]
       & & C_2'  \ar@{}[dll]|(0.50){(\gamma_2)} \ar@{}[drr]|(0.50){(\delta_2)}  \ar[d] \ar[ll] \ar[rr] && G_2' \ar[d] & \dots &  G_n' \ar[d]
       & & C_n' \ar@{}[dll]|(0.50){(\gamma_n)} \ar@{}[drr]|(0.50){(\delta_n)}  \ar[d] \ar[ll] \ar[rr] && G_n' \ar[d] \\
     C_0 \ar[rr] 
       & & G_0 & & C_1 \ar[ll] \ar[rr] & & G_1 & & C_2 \ar[ll] \ar[rr] & & G_2 & \dots &  G_n & & C_n \ar[ll] \ar[rr] & & G_n \\
       & & & & & & & & & & &  \\
       & & & & & & & & & &
       & J' 
           \ar[d] \ar[uuulll]|(0.3){} \ar@(dl,d) [uuulllllll]| (0.3) {} \ar[uuurrr]|(0.3){} \ar@{}[uulll]|(0.5){(\epsilon_2)} \ar@{}[uurrr]|(0.5){(\epsilon_n)} &
       & & & & & \\
       & & & & & & & & & &
       & C_0 \ar[uuulll]|(0.3){} \ar@(dl,d)[uuulllllll]|(0.56){} \ar[uuurrr]|(0.3){}  \ar@{}[uuullllllll]|(0.5){(\epsilon_1)} & & J \ar[ll] & & & &
      }
    } \]
  each of the matches $L_i \to G'_i \to G_i$ is convex.

  If all the critical pairs in $\mathcal R$ are path-joinable, then there exists a path-extension $p: G'_0 \to P_0$ that path-covers $G_0' \to G_0$ and $(P_{1,1} \leftarrow I) \LDPOstep{} (P_{0} \leftarrow I) \DPOstep{} (P_{1,2} \leftarrow I)$ is joinable. First, we can apply Lemma~\ref{lem:path-ext-rewrite} to translate a convex rewrite $(P_0 \leftarrow I) \DPOstep{} (P_{1,i} \leftarrow I)$ to a convex rewrite $(G'_0 \leftarrow J) \DPOstep{} (G'_{1,i} \leftarrow J)$. We can then apply Lemma~\ref{lem:path-ext-rewrite2} to extend this to a convex rewrite $(G_0 \leftarrow J) \DPOstep{} (G_{1,i} \leftarrow J)$. This yields a path-extension $q: G'_{1,i} \to P_{2,i}$ that path-covers $G'_{1,i} \to G_{1,i}$, hence we can iterate this process to get a convex rewrite $(G_{1,i} \leftarrow J) \DPOstep{} (G_{2,i} \leftarrow J)$, and so on.

  When we path-join the critical pair, we obtain $P_{m,1} \cong P_{n,2}$. If we remove all of the $\mathcal P$ hyperedges (and nodes connected only to $\mathcal P$ hyperedges), this will restrict to an isomorphism $G'_{m,1} \cong G'_{n,2}$, which in turn yields an isomorphism $G_{m,1} \cong G_{n,2}$. Hence the branching $(G_{1,1} \leftarrow J) \LDPOstep{} (G_0 \leftarrow J) \DPOstep{} (G_{1,2} \leftarrow J)$ is joinable by convex rewriting.
\end{proof}

\begin{definition}\label{def:max-path}
Given ma-rules $L_1 \leftarrow K_1 \rightarrow R_1$ and $L_2 \leftarrow K_2 \rightarrow R_2$ and an ma-pre-critical pair
$(G'_0 \leftarrow J)$, a maximal path relation is a binary relation $R \subseteq \out{G'_0} \times \inp{G'_0}$ such that
  \begin{enumerate}
    \item a mono $m : G'_0 \to H$ exists for an ma-hypergraph $H$ with $R_m = R$,
    \item the induced matchings $L_1 \to G'_0 \to H$ and $L_2 \to G'_0 \to H$ are convex, and
    \item no relation satisfying (1) and (2) is a proper superset of $R$.
  \end{enumerate}
\end{definition}

\begin{definition}\label{def:path-joinable}
Given ma-rules $L_1 \leftarrow K_1 \rightarrow R_1$ and $L_2 \leftarrow K_2 \rightarrow R_2$, an ma-pre-critical pair
$(G'_{1,1} \leftarrow J') \LDPOstep{} (G'_0 \leftarrow J') \DPOstep{} (G'_{1,2} \leftarrow J')$ is called \textit{path joinable}
if for any maximal path relation $R$, there exists a path extension $p : G'_0 \to P$ with $R_p = R$ such that the branching
$(P_{1,1} \leftarrow K) \LDPOstep{} (P_0 \leftarrow K) \DPOstep{} (P_{1,2} \leftarrow K)$ obtained by lifting the two rewriting steps 
in the critical pair along $p$ is joinable by convex rewriting.
\end{definition}

\begin{theorem}\label{thm:convex-critical-pair}
  Let $\mathcal R$ be a convex DPOI rewriting system. If all ma-pre-critical pairs are path joinable, then $\mathcal R$ is locally confluent.
\end{theorem}

\begin{proof}
  For any branching $(G_{1,1} \leftarrow J) \LDPOstep{} (G_0 \leftarrow J) \DPOstep{} (G_{1,2} \leftarrow J)$, we can find an ma-pre-critical pair based at $G_0'$ where the embedding $e : G'_0 \to G_0$ has a path relation $R_e$ satisfying conditions (1) and (2) in Definition~\ref{def:max-path}. Hence, there exists a path-extension $p : G'_0 \to P_0$ where $R_e \subseteq R_p$ and the associated branching is joinable by convex rewriting. Hence by Lemma~\ref{lem:lift-joins}, the branching based at $G_0$ is also joinable by convex rewriting. Therefore $\mathcal R$ is locally confluent.
\end{proof}

The converse of this theorem is almost true, but with a small caveat that one needs to consider local confluence of ma-hypergraphs over the full signature $\Sigma + \mathcal P$ containing the formal path generators, rather than just $\Sigma$.

\begin{theorem}
  Let $\mathcal R$ be a convex DPOI rewriting system. If it is locally confluent for all ma-hypergraphs labelled by $\Sigma + \mathcal P$, then all ma-pre-critical pairs are path joinable.
\end{theorem}

The proof is immediate, since failing to join the path-extension of a ma-pre-critical pair witnesses a failure of local confluence. There is no reason \textit{a priori} that the above theorem would hold just for ma-hypergraphs. Hence, one can see the inclusion of the formal path generators as a sort of `stabilisation' of the theory that rules out certain degenerate cases of convex rewriting systems, such as those where certain paths never exist or can always be broken by rewriting.

Theorem~\ref{thm:convex-critical-pair} gives us an effective way to check local confluence. For an ma-pre-critical pair, we need to enumerate all the maximal path relations, and for each one, construct a path extension and check if it is joinable. While there could in principle be exponentially many of these for each critical pair, conditions (1)-(3) in Definition~\ref{def:max-path} rule many of them out. We will see this process in action in the case study in Section~\ref{sec:non-left-conn}.

\section{Case studies}
\label{case studies}
We close the paper by providing two positive examples of our confluence results. Both of them
concern ma-hypergraphs, distinguishing between left-connected systems and convex rewriting.

\subsection{Left-connected and confluent: non-commutative bimonoids}\label{sec:confluenceBimonoids}
First case study is an application of the results on left-connected systems, showing confluence of the theory $\NCB{}$ of non-commutative bimonoids 
(Example \ref{exm:props}\ref{ex:ncbialgebras}). Below is the interpretation of the theory as a DPO system 
$\allTosem{\rewiring{\mathcal{R}_{\textbf{NBiM}}}}$, which was shown to be terminating in~\cite{BGKSZ-parttwo}
\[
    \begin{array}{rrclrrcl}
          \NCB{1} := & \!\!\!\!\cgr{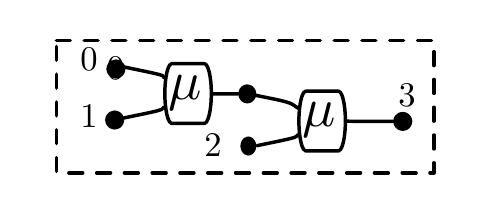}\!\!\!& \!\!\!\leftarrow \!\!\!\cgr{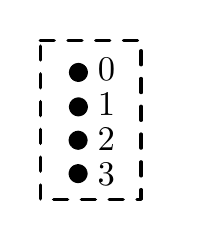} \!\!\!\!\! \rightarrow & \!\!\!\!\!\cgr{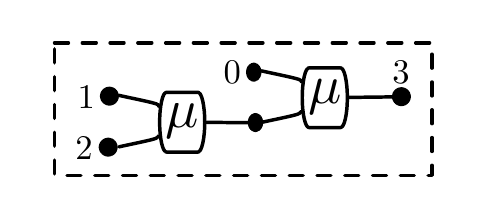}
          &
              \NCB{2} := & \!\!\!\!\cgr{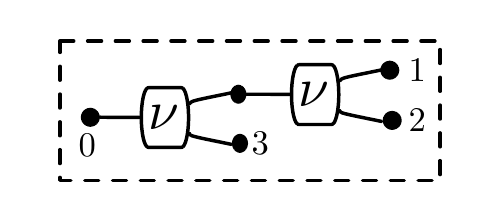}\!\!\!\! & \!\!\leftarrow \!\!\! \cgr{fourDotsFramed.pdf} \!\!\!\!\!  \rightarrow\! & \!\!\!\!\!\!\cgr{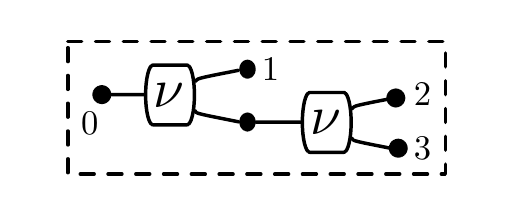}\!\!\!\!
              \\  
      \NCB{3} := & \!\!\!\!\cgr{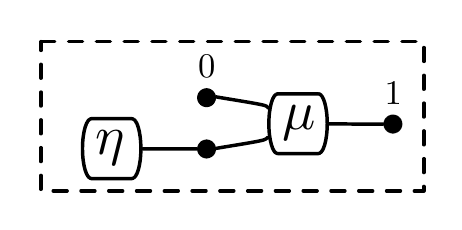}\!\!\! & \leftarrow \!\!\!\!\cgr{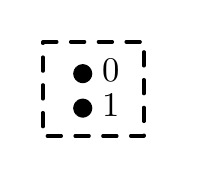}\!\!\!\!\! \rightarrow & \!\!\!\!\cgr{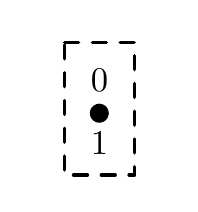}\!\!\!\!
      &
      \NCB{4} := & \!\!\!\!\cgr{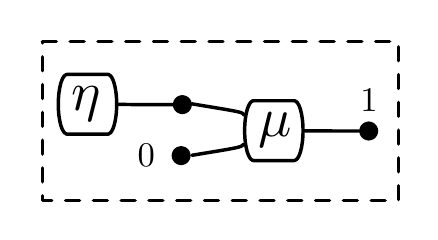}\!\!\! & \leftarrow \!\!\!\!\cgr{twoDotsFramed.pdf}\!\!\!\!\! \rightarrow & \!\!\!\!\cgr{oneDotZeroOne.pdf}\!\!\!\!
      \\
      \NCB{5} := & \!\!\!\!\cgr{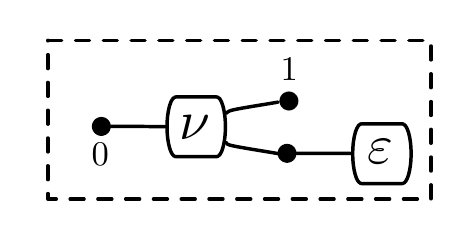}\!\!\! & \leftarrow \!\!\!\!\cgr{twoDotsFramed.pdf}\!\!\!\!\! \rightarrow & \!\!\!\!\cgr{oneDotZeroOne.pdf}\!\!\!\!
      &
      \NCB{6} := &  \!\!\!\!\cgr{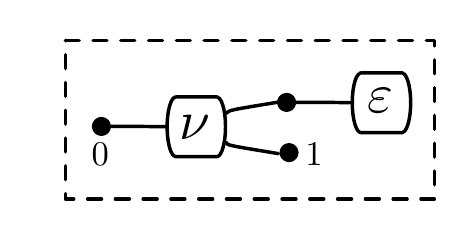}\!\!\! & \leftarrow \!\!\!\!\cgr{twoDotsFramed.pdf}\!\!\!\!\! \rightarrow & \!\!\!\!\cgr{oneDotZeroOne.pdf}\!\!\!\!
      \\
      \NCB{7} := & \!\!\!\!\cgr{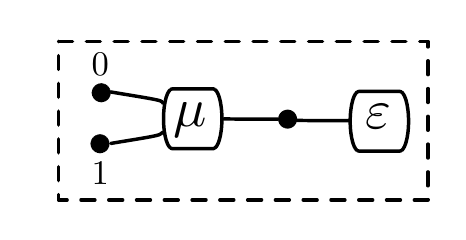}\!\!\! & \leftarrow \!\!\!\!\cgr{twoDotsFramed.pdf}\!\!\!\!\! \rightarrow & \!\!\!\!\!\cgr{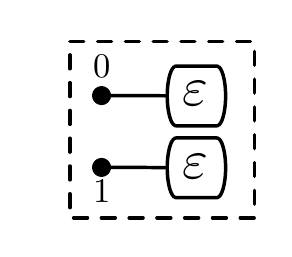}\!\!\!\!
      &
      \NCB{8} := & \!\!\!\!\cgr{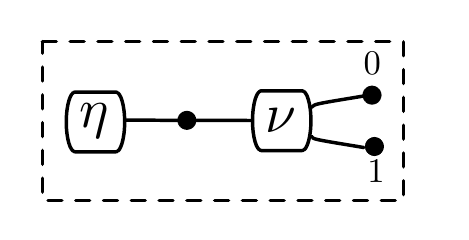}\!\!\!\!\! & \leftarrow \!\!\!\!\cgr{twoDotsFramed.pdf}\!\!\!\!\! \rightarrow & \!\!\!\!\cgr{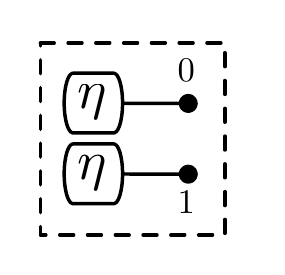}\!\!\!\!
      \\
      \NCB{9} := & \!\!\!\!\cgr{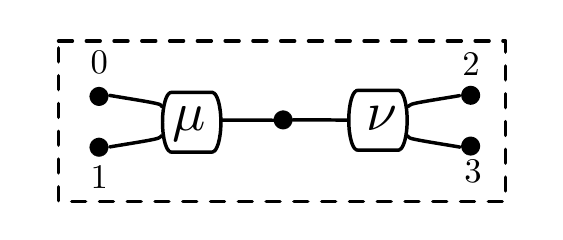}\!\!\!\!\!\! & \leftarrow \!\!\!\! \cgr{fourDotsFramed.pdf} \!\!\!\!\! \rightarrow & \!\!\!\!\!\!\cgr{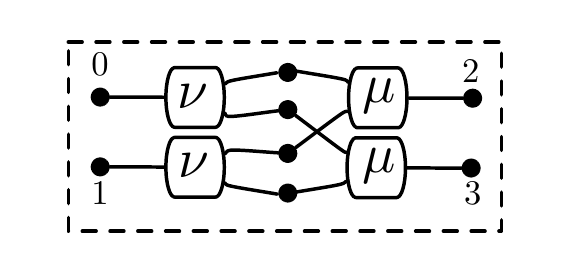}\!\!
      &
      \NCB{10} := & \!\!\!\!\cgr{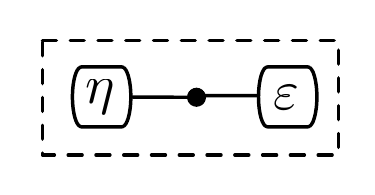}\!\!\! & \leftarrow \!\!\! \cgr{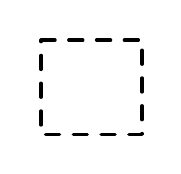} \!\!\! \rightarrow &  { \!\!\!\!\cgr{emptyFrame.pdf}\!\!\!\!}
    \end{array}
  \]
Given that the system is terminating, it suffices to show local confluence. Observe that $\allTosem{\rewiring{\mathcal{R}_{\textbf{NBiM}}}}$ is left-connected: monogamy is ensured by the fact that it is in the image of $\allTosem{\rewiring{\cdot}}$; strong connectedness and left-linearity hold by inspection of the set of rules. We can thus use Theorem~\ref{th:locconflwellformed} and local confluence follows from joinability of the ma-pre-critical pairs. Among them, the pairs without overlap of rule applications pose no problem: they are trivially joinable. One example is given below, with the middle grey graph acting as the interface for all depicted derivation steps

\[ \cgr[width=11.5cm]{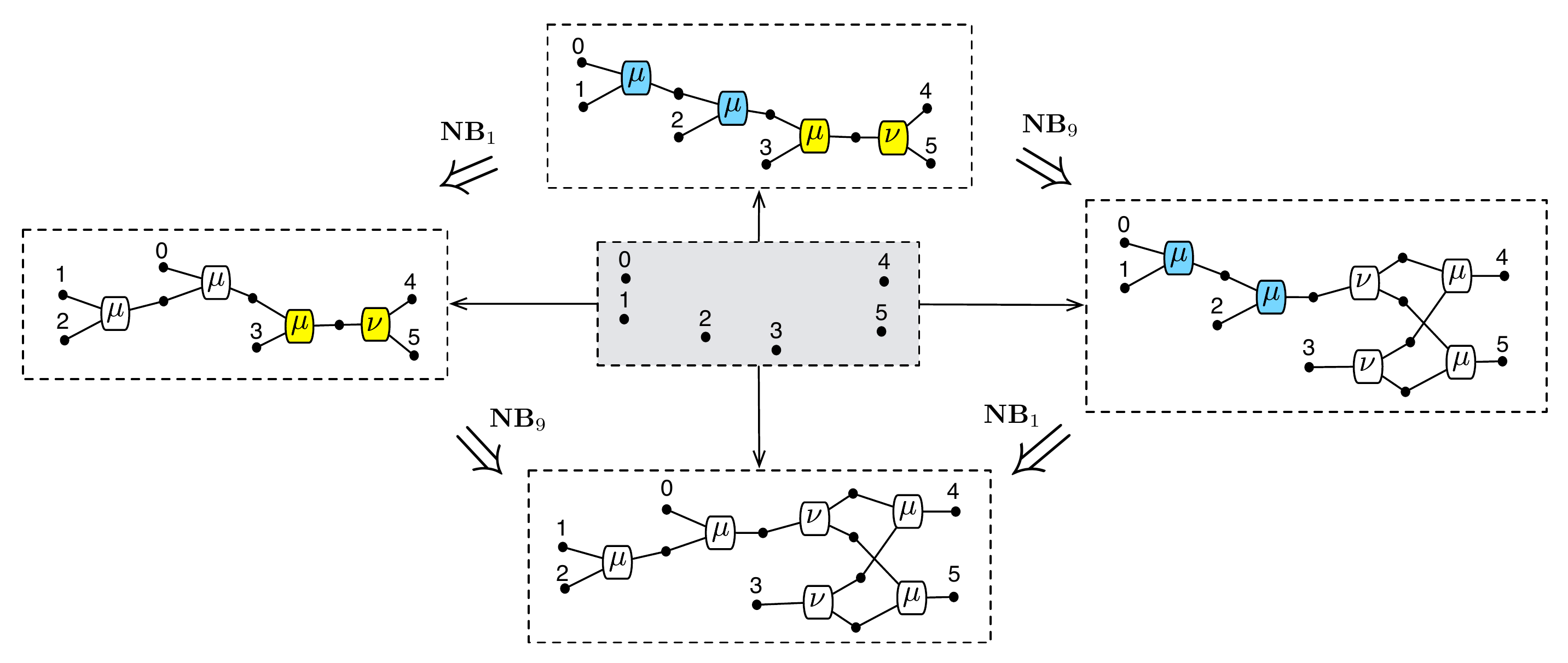} \]
Thus we confine ourselves to analysing actual critical pairs, with overlapping rule applications. One such pair is given below, also involving rules $\NCB{1}$ and $\NCB{9}$. Again, we show how it is joined, with the interface of each step drawn in the centre
\[ \cgr{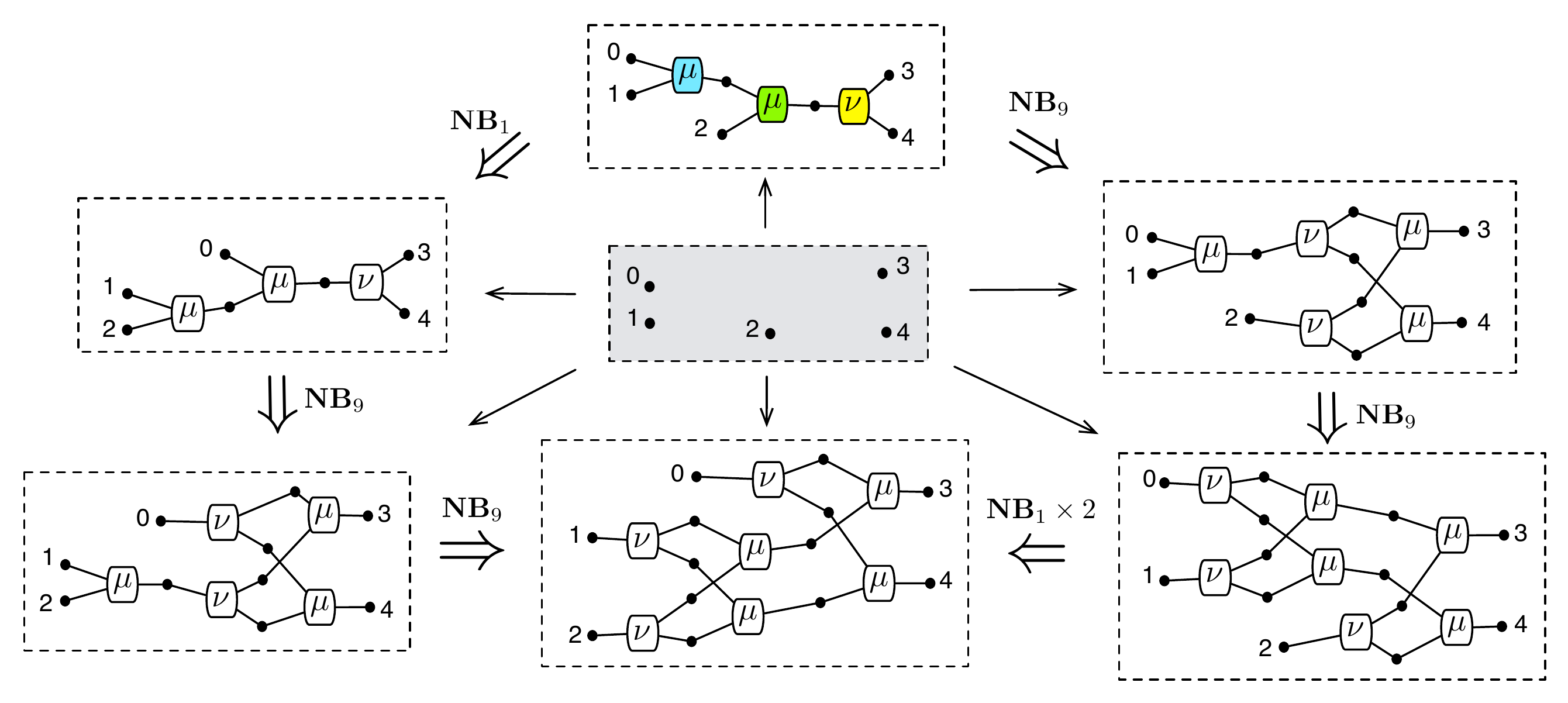} \]
Overall there are 22 critical pairs to consider. For each of them we only show the graph exhibiting the overlap. It is straightforward to check that the corresponding pairs are all joinable
\[ \cgr[width=12cm]{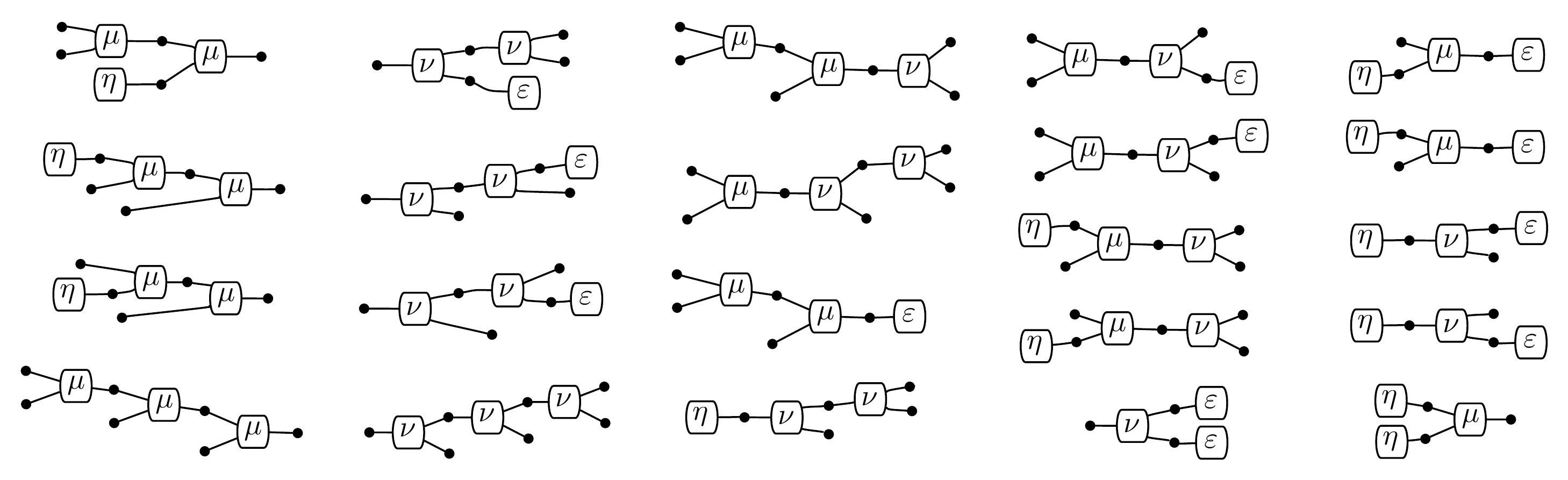} \]
We can thereby conclude that $\NCB{}$ is a confluent rewriting system. Since it is also terminating, equivalence of terms in $\NCB{}$ is decidable by means of rewriting. Note that, by virtue of Corollary \ref{cor:confluenceSMTs}, the above pre-critical pair analysis can be automated.

\subsection{A confluent, non-left-connected example}\label{sec:non-left-conn}
\label{nlc confluence}
We now consider an example of a rewriting system that is not left-connected, and demonstrate a proof of confluence by means of path extensions. Let $\Sigma = \{ f : 2 \to 2, g : 1 \to 0, h : 0 \to 1 \}$, satisfying one equation
\[ \scalebox{0.75}{\input{./figures/case-eq.tikz}} \]
which translates into the following ma-rule
\[ \NLC{1} := \scalebox{0.75}{\input{./figures/case-rule.tikz}} \]
The rule strictly decreases the number of $f$-labelled hyperedges in a graph, so \NLC{} is clearly terminating. Hence, it suffices to check local confluence to prove that
 \NLC{} is confluent.

The rule \NLC{} has two types of ma-pre-critical pairs. The first type is a genuine critical pair
\begin{equation}\label{eq:case-cp1}
  \scalebox{0.75}{\input{./figures/case-crit-pair.tikz}}
\end{equation}
and the second is a parallel pair
\begin{equation}\label{eq:case-cp2}
  \scalebox{0.75}{\input{./figures/case-pre-crit-pair.tikz}}
\end{equation}

All of the other ma-pre-critical pairs are variations of \eqref{eq:case-cp1} and \eqref{eq:case-cp2}, obtained by joining some outputs to some inputs in such a way that the two matchings of \NLC{1} remain convex.

The maximal path relations for \eqref{eq:case-cp1} and \eqref{eq:case-cp2} can be computed by exhaustive enumeration (see Appendix~\ref{app:extensions}). For \eqref{eq:case-cp1}, there is only one maximal path relation $R_p = \{ (7, 3) \}$. Hence, we can form the extension by adjoining a $\mathcal P$-hyperedge that creates a path from output 7 to input 3. This is then joinable
\[ \scalebox{0.75}{\input{./figures/case-crit-path-ext-join.tikz}} \]

The ma-pre-critical pair~\eqref{eq:case-cp2} has three maximal path relations
\begin{align*}
  R_p & = \{ 6, 7, 8 \} \times \{ 3, 4, 5 \} \\
  R_q & = \{ 9, 10, 11 \} \times \{ 0, 1, 2 \} \\
  R_r & = \{ (11, 2), (8, 5) \} \\
\end{align*}

We can construct path extensions for each of these as follows
\begin{align*}
  R_p & :\ \ \scalebox{0.75}{\input{./figures/case-pre-crit-ext1.tikz}} \\
  R_q & :\ \ \scalebox{0.75}{\input{./figures/case-pre-crit-ext2.tikz}} \\
  R_r & :\ \ \scalebox{0.75}{\input{./figures/case-pre-crit-ext3.tikz}} \\
\end{align*}

Each of these is joinable by convex rewriting. The 24 variations of \eqref{eq:case-cp1} and \eqref{eq:case-cp2} are all basically identical to these two cases, with the only difference being that one or more paths in the path extension is replaced by an output connected directly to an input. These are also all joinable, and hence the system \NLC{} is confluent.

\section{Conclusions}
\label{conclusions}

The starting observation of this paper (Theorem~\ref{th:locconfl}) is that the Knuth-Bendix property holds for DPOI rewriting.
As a consequence (Corollary~\ref{cor:comput}), confluence is decidable for terminating systems. 

The relevance of this is two-fold. On the conceptual side, it puts graph rewriting in tight correspondence with term rewriting: 
When considering rewriting with interfaces, confluence is decidable for graphs as it is for terms~\cite{knuth1970simple}, 
while the appropriate notion of ground confluence is undecidable in both cases~\cite{kapur1990ground,Plump1993}.

On the side of applications, our result allows one to study confluence for string diagrams.
One consequence of Theorem~\ref{th:locconfl} and of our previous work in~\cite{BonchiGKSZ_lics16} is that, for all those 
symmetric monoidal theories including a special Frobenius structure, which are already commonplace in computer science~\cite{Bruni2006,Bruni2011,Sobocinski2014,Coecke2008,Coecke2012,Bonchi2015,BaezErbele-CategoriesInControl,Fong2015}, 
local confluence can be checked by means of critical pair analysis. 
Moreover, confluence can be decided automatically when termination is guaranteed (Corollary~\ref{cor:confluenceFROB}). 
Analogous results on critical pairs hold for those theories that do not include a special Frobenius structure, albeit with 
a few caveats.
More precisely, confluence is accomplished by two kinds of restrictions. The first choice is to curb the family of admissible rules  
to left-connected systems (Theorem~\ref{th:locconflwellformed}): The notion of critical pair is substantially unchanged,
so that confluence can be decided automatically along the same lines as in the Frobenius case (Corollary~\ref{cor:confluenceSMTs}). 
The second option is to restrict the family of admissible rewriting steps to convex matches, at the cost of checking
a larger family of critical pairs in order to include path extensions (Theorem~\ref{thm:convex-critical-pair}). The family is still computable, 
as witnessed by the algorithm proposed in Appendix~\ref{app:extensions}, thus obtaining again confluence decidability for terminating systems.

Our results apply to a variety of other non-Frobenius theories, such as those in~\cite{Lafont2003,Ghica13,Fiore2013}. 
In any case, in  all the proposed scenarios, these decision procedures are amenable to implementation in string diagram 
rewriting tools like Quantomatic 
\cite{Kissinger_quantomatic} (via an encoding of hypergraphs) or directly in hypergraph-based rewriting tools.


\bibliographystyle{plain}
\bibliography{catBib3Rev}

\newpage

\appendix

\section{Competing Interest Declaration}
Competing interests: The authors declare none.

\section{Enumeration of Maximal Path Relations}\label{app:extensions}

The case study in section~\ref{sec:non-left-conn} makes use of the following Python code for enumerating maximal path relations. The main function \texttt{find\_extensions} takes as arguments
\begin{itemize}
  \item \texttt{inputs}: a set of input vertices,
  \item \texttt{outputs}: a set of output vertices,
  \item \texttt{paths}: a set of pairs indicating a path connects the given input to the output in the pre-critical pair,
  \item \texttt{forbidden\_paths}: a set of pairs indicating that a given path must \textit{not} exist from an output to an input, due to convexity.
\end{itemize}

\begin{lstlisting}
from itertools import chain, combinations

# iterator over the powerset of a set, ordered from largest
# to smallest subsets
def powerset(iterable):
    s = list(iterable)
    return map(lambda it: set(it),
            chain.from_iterable(combinations(s, r)
                for r in range(len(s)+2,-1,-1)))

def find_extensions(inputs, outputs, paths, forbidden_paths):
    paths_rev = set(map(lambda p: (p[1],p[0]), paths))

    forbidden = (forbidden_paths
                 .union({(i,i) for i in inputs})
                 .union({(i,i) for i in outputs})
                 .union(paths_rev))

    # returns true of the transitive closure of (c UNION paths)
    # contains no forbidden pairs
    def is_sln(c):
        c = c.union(paths)
        while True:
            d = c.copy()
            for (x,y0) in c:
                for (y1, z) in c:
                    if y0 == y1: d.add((x,z))
            if len(d.intersection(forbidden)) > 0:
                return False
            else:
                if len(c) == len(d):
                    return True
                else:
                    c = d

    candidates = set((i,j) for i in outputs for j in inputs) - forbidden

    slns = []
    sz = len(candidates) + 1
    for c in powerset(candidates):
        if len(c) < sz:
            print("%s solutions of size >= %s" % (len(slns), sz))
            sz = len(c)
        if is_sln(c) and not any(c <= s for s in slns):
            slns.append(c)
    return slns

######################################################################

if __name__ == "__main__":
    exts1 = find_extensions(
        inputs = { 0, 1, 2, 3 },
        outputs = { 4, 5, 6, 7 },
        paths = {
            (0,4),(0,5),(0,6),(0,7),
            (1,4),(1,5),(1,6),(1,7),
            (2,4),(2,5),(2,6),
            (3,4),(3,5)
            },
        forbidden_paths = { (7,2), (6,3) }
    )

    print("\n\nCP1 got %s extensions:"% len(exts1))
    for e in exts1: print(e)

    exts2 = find_extensions(
        inputs = { 0, 1, 2, 3, 4, 5 },
        outputs = { 6, 7, 8, 9, 10, 11 },
        paths = {
            (0,6),(0,7),(0,8),
            (1,6),(1,7),(1,8),
            (2,6),(2,7),
            (3,9),(3,10),(3,11),
            (4,9),(4,10),(4,11),
            (5,9),(5,10)
            },
        forbidden_paths = { (8,2), (11,5) }
    )

    print("\n\nCP2 got %s extensions:"% len(exts2))
    for e in exts2: print(e)
\end{lstlisting}
\end{document}